\documentclass[11pt,letterpaper,oneside,english]{article}
\usepackage[top=3.2cm, bottom=3.2cm, left=3.2cm, right=3.2cm]{geometry}

\usepackage{amssymb,amsbsy,latexsym}
\usepackage{amsmath}
\usepackage{graphics, subfigure, float}
\usepackage{fp, calc}
\usepackage{bm}
\usepackage[latin1]{inputenc}
\usepackage[american]{babel}
\usepackage[T1]{fontenc} % Important: Use before loading the fonts

\def\url#1{{\texttt #1}}

\usepackage{bm}

\usepackage{amscd,amsthm}

\usepackage{pst-all}
\usepackage{verbatim, comment}

\newtheoremstyle{theorem}{1em}{1em}{\slshape}{0pt}{\bfseries}{.}{ }{}
\theoremstyle{theorem}
\newtheorem{theorem}{Theorem}
\newtheorem*{theorem*}{Theorem}
\newtheorem{corollary}[theorem]{Corollary}

\newtheorem{lemma}[theorem]{Lemma}

\theoremstyle{remark}

\newtheorem*{remark*}{Remark}

\providecommand{\setN}{\mathbb{N}}
\providecommand{\setZ}{\mathbb{Z}}

\providecommand{\setR}{\mathbb{R}}

\newcommand{\conv}{\textrm{conv}}
\newcommand{\rk}{\textrm{rk}}
\newcommand{\xc}{\textrm{xc}}
\newcommand{\E}{\mathop{\mathbb{E}}}

 \theoremstyle{theorem}
\newtheorem*{homework*}{Homework}
 \theoremstyle{definition}
 \newtheorem{definition}{Definition}
\newtheorem*{claim*}{Claim}

\newtheorem*{example*}{Beispiel}

        \def\drawRect#1#2#3#4#5{
           \FPeval{\x2}{(#2) + #4} 
           \FPeval{\y2}{(#3) + #5} 
           \pspolygon[#1](#2,#3)(\x2,#3)(\x2,\y2)(#2,\y2)
        }

\def\placeIIID#1#2#3#4{
           \FPeval{\x}{(#1) + 0.3*(#2)} 
           \FPeval{\y}{0.5*(#2) + (#3)} 
           \rput[c](\x,\y){#4}
}

\usepackage[displaymath,textmath,graphics, subfigure, floats]{preview} %sections,
\PreviewEnvironment{myproblem}
\PreviewEnvironment{defn} 
\PreviewEnvironment{center} 
\PreviewEnvironment{pspicture} 
 \usepackage{datetime}

\makeatother

\begin{document}

\title{The matching polytope has exponential extension complexity}
\date{} %\date{\today, \currenttime}
\author{Thomas Rothvo{ss}\thanks{Email: {\tt rothvoss@uw.edu}. Part of the work was done while
the author was visiting the University of Waterloo and the rest was done while the author was a PostDoc at the Massachusetts Institute of Technology. Research supported in part by ONR grant N00014-11-1-0053 and by NSF contract 1115849. This is the journal version of the paper with the same title that appeared at STOC'14.} \vspace{2mm} \\ University of Washington, Seattle} %%Massachusetts Institute of Technology}
\maketitle

\begin{abstract}
A popular method in combinatorial optimization is to express polytopes $P$, 
which may potentially have exponentially many facets,
as solutions of linear programs that use few extra variables to reduce the number of constraints
down to a polynomial. 
After two decades of standstill, recent years have brought amazing progress in showing lower bounds for
the so called \emph{extension complexity}, which for a polytope $P$ denotes the smallest 
number of inequalities necessary to describe a higher dimensional polytope $Q$ that can
be linearly projected on $P$.
 
However, the central question in this field remained wide
open: can the \emph{perfect matching polytope} be written as an LP with polynomially many constraints?

We answer this question negatively. In fact, the extension complexity of the perfect matching polytope
in a complete $n$-node graph is $2^{\Omega(n)}$. By a known reduction this also improves the lower bound on 
the extension complexity for the TSP polytope from $2^{\Omega(\sqrt{n})}$ to $2^{\Omega(n)}$. 
\end{abstract}
% As known methods like the rectangle covering lower bound are known
% to not yield superpolynomial lower bounds, 
% From a technical point of view, 
% A linear extension of a polytope is a higher-dimensional polytope
% The field of extended formulations deals with expressing polytopes that typically have exponentially
% many facets as solutions to linear programs with polynomially many extra variables
% The matching polytope is the convex hull of all perfect matchings in a complete $n$-node graph. 

\section{Introduction}

Linear programs are at the heart of combinatorial optimization as they allow to model
a large class of polynomial time solvable problems such as flows, matchings and matroids.  The concept of
LP duality lead in many cases to structural insights that in turn lead to specialized
polynomial time algorithms. In practice, general LP solvers turn out to be
very competitive for many problems, even in cases in which specialized algorithms have 
better theoretical running time. Hence it is particularly interesting to model problems 
with as few linear constraints as possible. For example if we consider the convex hull $P_{ST}$ of the
characteristic vectors of all spanning trees in a complete $n$-node graph, then this 
polytope has $2^{\Omega(n)}$ many facets~\cite{MatroidAndTheGreedyAlgo-Edmonds-1971}. 
However, one can write $P_{ST} = \{ x \mid \exists y: (x,y) \in Q\}$ with a
 higher dimensional polytope $Q$ with only $O(n^3)$ many inequalities~\cite{CF4SpanningTree-KippMartin1991}. 
Hence, instead of optimizing
a linear function over $P_{ST}$, one can optimize over $Q$. In fact, $Q$ is called 
a \emph{linear extension} of $P_{ST}$ and the minimum number of facets of any linear extension is called 
the \emph{extension complexity} and it is denoted by $\xc(P_{ST})$; in this case $\xc(P_{ST}) \leq O(n^3)$.
If $\xc(P)$ is bounded by a polynomial in $n$, then we say that $P \subseteq \setR^n$ has a \emph{compact formulation}.

Other examples of non-trivial compact formulations contain the permutahedron~\cite{PermutahedronNlogNformulation_Goemans10}, the parity polytope, 
the matching polytope in planar graphs~\cite{CompactForm4MatchingPolytopeInPlanarGraphs-Baharona1993} 
and more generally the matching polytope in graphs with bounded genus~\cite{CompactLPforPerfectMatchingOnBoundedGenusSurfacesGerards1991}.

A natural question that emerges is which polytopes do \emph{not} admit a compact formulation. 
The first progress was made by Yannakakis~\cite{ExpressingCOproblemsAsLPs-Yannakakis1991} 
who showed that any \emph{symmetric} extended formulation
for the matching polytope and the TSP polytope must have exponential size. Conveniently, this 
allowed to reject a sequence of flawed ${\bf{P}}={\bf{NP}}$ proofs, which claimed to have (complicated) polynomial
size LPs for TSP.
It was not clear a priori whether the symmetry condition would be essential, but 
Kaibel, Pashkovich and Theis~\cite{SymmetryMatters-KaibelPashkovichTheis-IPCO2010} showed that for the
convex hull of all $\log n$-size matchings, there is a compact asymmetric formulation, but
no symmetric one. 
Using a counting argument the author of this paper was able to prove that a random
 0/1 polytope would have  extension complexity that is exponential in the dimension~\cite{Some01PolytopesNeedExpSizeLPs-Rothvoss-MathProg2013}, 
but without being able to show this
for a concrete polytope. In fact, this existential technique was extended to polygons by 
Fiorini, Rothvoss and Tiwary~\cite{ExtendedFormulations-for-Polygons-FioriniRothvossTiwary-DCG2012} and 
to the SDP extension complexity by~Bri{\"e}t, Dadush and Pokutta~\cite{Existence-01polytopes-with-high-SDP-rank-BrietDadushPokutta-ESA13}.

The major breakthrough by Fiorini, Massar, Pokutta, Tiwary and de Wolf~\cite{TSP-lower-bound-Fiorini-et-al-STOC2012} showed that
several well studied polytopes, including the correlation polytope and the TSP polytope, have
exponential extension complexity. More precisely, they show that the 
\emph{rectangle covering lower bound}~\cite{ExpressingCOproblemsAsLPs-Yannakakis1991} for 
the correlation polytope is exponential, for which
they use known tools from communication complexity~such as 
Razborov's \emph{rectangle corruption lemma}~\cite{Disjointness-Razborov90}. 
Their method was extended by Braun, Fiorini, Pokutta and Steurer~\cite{Approximate-LPs-BraunFioriniPokuttaSteurer-FOCS12}
to show also lower bounds on the size of approximations to the correlation polytope; via a reduction this
then implies that any polynomial size LP for the convex hull of all cliques in all $n$-node graphs must
have an integrality gap of $n^{1/2-\varepsilon}$. This quantity was subsequently improved 
to~$n^{1-\varepsilon}$ by Braverman and Moitra~\cite{InfoComplexity-approach-to-extended-formulations-BravermanMoitra-STOC13}
using an information theory approach; the same bound was later provided by Braun and 
Pokutta~\cite{CommonInfo-and-Unique-Disjointness-BraunPokutta-FOCS2013}
using the framework of common information.

One insight that appeared already in \cite{ExpressingCOproblemsAsLPs-Yannakakis1991,TSP-lower-bound-Fiorini-et-al-STOC2012}
is that if a ``hard'' polytope $P$ is the linear projection of a face of another polytope $P'$, then $\xc(P') \geq \xc(P)$.
This way, the ``hardness'' of the correlation polytope can be translated to many other polytopes using a reduction
(in fact, in many cases, the usual ${\bf NP}$-hardness reduction can be used);
see \cite{ExtensionComplexity-of-Knapsack-Polytope-PokuttaVanVyve-ORL13,ExtensionComplexity-of-combinatorial-polytopes-AvisTiwaryICALP13} for some examples.

A completely independent line of research was given by Chan, Lee, Raghavendra and 
Steurer~\cite{LPs-for-CSPs-ChanLeeRaghavendraSteurer-FOCS2013} who use techniques 
from Fourier analysis to show that for constraint satisfaction problems, known integrality gaps for the 
Sherali-Adams LP translate to lower bounds for any LP of a certain size. For example they show that no LP of size
$n^{O(\log n / \log \log n)}$ can approximate MaxCut better than $2-\varepsilon$. This is particularly interesting
as in contrast the gap of a polynomial size SDP relaxation is around $1.13$, as was shown by Goemans and Williamson~\cite{MaxCut-GoemansWilliamsonJACM95}.

%Finally, we want to remark 
%In contrast, very little is known on lower bounds for SDP, just that there must be \cite{Existence-01polytopes-with-high-SDP-rank-BrietDadushPokutta-ESA13}.
However, all those polytopes model  ${\bf{NP}}$-hard problems and naturally, no complete description of their
facets is known (and no efficiently separable description is possible if ${\bf NP} \neq {\bf coNP}$). 
So what about nicely structured combinatorial polytopes 
that admit polynomial time algorithms to optimize linear functions? % (and consequently also polynomial time separation oracles)?
The most prominent example here is the \emph{perfect matching polytope} $P_{PM}$, which is the convex hull of all characteristic
vectors of perfect matchings in a complete $n$-node graph $G=(V,E)$. 
A well-known work of Edmonds~\cite{matchingPolytop-Edmonds1965} shows that 
apart from requiring non-negativity, the degree-constraints plus the \emph{odd-cut inequalities} are enough 
for an inequality description. In other words, we can write
\begin{eqnarray*}
P_{PM}(n) &=& \textrm{conv}\{ \chi_M \in \setR^{E}  \mid M \subseteq E\textrm{ is a perfect matching} \} \\
&=& \{ x \in \setR^E \mid x(\delta(v)) = 1 \; \forall v \in V; \; x(\delta(U)) \geq 1 \; \forall U \subseteq V: \; |U| \textrm{ odd}; \; x_e \geq 0 \; \forall e \in E \}.
\end{eqnarray*}
Here, $\delta(U)$ denotes the set of edges that have exactly one endpoint in $U$. 
We omit the number $n$ of nodes if they are clear from the context.
Note that there are only $n$ degree constraints and $O(n^2)$ non-negativity constraints, but $2^{\Omega(n)}$
odd set inequalities. Any linear function can be optimized over $P_{PM}$ in 
strongly polynomial time using Edmonds' algorithm~\cite{matchingPolytop-Edmonds1965}. 
Moreover, given any point $x \notin P_{PM}$, a violating inequality
can be found in polynomial time via the equivalence of optimization and separation or 
using Gomory-Hu trees, see Padberg and Rao~\cite{MinimumOddCuts-PadbergRao1982}.
%The structure of the polytope is well studied, for example vertices $\chi_{M}$ and $\chi_{M'}$ are neighbours on the
%1-skeleton if and only of their symmetric difference consists of exactly one even length cycle. 
There are compact formulations for $P_{PM}$ for special graph classes~\cite{CompactLPforPerfectMatchingOnBoundedGenusSurfacesGerards1991}
and every \emph{active cone} of $P_{PM}$ admits a compact formulation~\cite{VenturaEisenbrandORL03}.
Moreover, the best known upper bound on the extension complexity in general graphs is 
$\textrm{poly}(n) \cdot 2^{n/2}$~\cite{ExtendedForm-and-Rand-Protocolls-FaenzaFioriniGrappeTiwary-ISCO12}, 
which follows from the fact that   $\textrm{poly}(n) \cdot 2^{n/2}$ many randomly taken complete bipartite
graphs cover all matchings and that the convex hull of the union of polytopes 
can be described with an extended formulation whose size is basically the sum of the individual 
number of inequalities~\cite{DisjunctiveProgramming-Balas85}.
Moreover,
\begin{eqnarray}
  P_{M} &=& \textrm{conv}\{ \chi_M \in \setR^E \mid M \subseteq E\textrm{ is a matching} \}  \label{eq:MatchingPolytope} \\
 &=& \{ x \in \setR^E \mid x(\delta(v)) \leq 1 \; \forall v \in V; \; x(E(U)) \leq \tfrac{|U|-1}{2} \; \forall U \subseteq V : |U| \textrm{ odd}; \; x_e \geq 0 \; \forall e \in E  \} \nonumber
\end{eqnarray}
is the convex hull of all matchings in $G$ (not just the perfect ones). Here, $E(U)$ denotes the edges
running inside of $U$.
Since $P_{PM}$ is a \emph{face} of $P_M$, we have $\xc(P_{PM}) \leq \xc(P_M)$. 
But one can also prove that $\xc(P_{M}(n)) \leq \xc(P_{PM}(2n))$\footnote{Simply take a complete graph $G = (V,E)$ on $2n$ nodes and a subset $U \subseteq V$ of $|U| = n$ nodes. Then if we take 
all perfect matchings $M$ in $G$, then $M \cap E(U)$ gives all matchings in $U$. This construction implies that $P_{M}(n)$ is the linear projection of a face of $P_{PM}(2n)$.}.
For a detailed discussion of the matching polytope we refer to the 
book of Schrijver~\cite{CombinatorialOptimizationABC-Schrijver}.

\subsection{Our contribution}

Despite of all those nice structural properties we show:
\begin{theorem} \label{thm:MainTheorem}
For all even $n$, the extension complexity of the perfect matching polytope in the complete $n$-node graph is $2^{\Omega(n)}$.
\end{theorem}
This answers a question that was open at least since the paper of 
Yannakakis~\cite{ExpressingCOproblemsAsLPs-Yannakakis1991}. The previously best known lower bound
was $\Omega(n^2)$~\cite{CombBoundsOnExtendedFormulations-FioriniKaibelEtAl-DiscMath13}.
As argued above, this also implies that $\xc(P_M) \geq 2^{\Omega(n)}$.

%Recall that the perfect matching polytope is a face of the matching polytope itself, hence the bound also holds
%for the convex hull of all (not necessarily perfect) matchings.
Yannakakis' paper~\cite{ExpressingCOproblemsAsLPs-Yannakakis1991} also describes a linear projection
from a face of the TSP polytope in an $O(n)$-node graph to the perfect matching polytope in an $n$-node
graph. This immediately implies a lower bound for TSP as well, which improves on the 
$2^{\Omega(\sqrt{n})}$ bound due to~\cite{TSP-lower-bound-Fiorini-et-al-STOC2012}.
\begin{corollary}
For all $n$, the convex hull $P_{\textrm{TSP}}$ of the characteristic vectors of all Hamiltonian cycles in a complete $n$-node
graph has extension complexity $2^{\Omega(n)}$.
\end{corollary}
Also this bound is tight up to constant factors in the exponent.
After the publication of the conference version of this paper, 
Braun and Pokutta~\cite{DBLP:journals/corr/BraunP14} extended our arguments to show that any polytope $K$
with $P_M \subseteq K \subseteq (1 + \frac{1}{2n}) P_M$ must still have 
extension complexity $2^{\Omega(n)}$, implying that there is no ``FPTAS-style''
approximate extended formulation for matching. We provide a simple reduction
which extends their result over the whole parameter range of $\varepsilon$ (see
Section~\ref{sec:Inapproximability}).
It seems this has not been observed before. 
\begin{corollary} 
%Let $P_M$ be the convex hull of 
%all matchings. 
Let $K$ be a polytope with $P_M \subseteq K \subseteq (1+\varepsilon) P_M$. 
Then $\xc(K) \geq 2^{\Omega(\min\{ 1/\varepsilon,n\})}$.
\end{corollary}

\section{Our approach}

Formally, the \emph{extension complexity} $\xc(P)$ is the smallest number of facets
of a (higher-dimensional) polyhedron $Q$ such that there is a linear projection $\pi$
with $\pi(Q) = P$. This definition seems to ignore the dimension, but one can always eliminate a non-trivial 
lineality space from $Q$ and make $Q$ full-dimensional, and then the dimension of $Q$ is bounded by the number of inequalities anyway.
Before we continue our discussion of the matching polytope, consider a general polytope 
$P$ and let $x_1,\ldots,x_v$ be a list of its vertices. Moreover, let $P = \conv\{x_1,\ldots,x_v\} = \{ x \in \setR^n \mid Ax \leq b\}$
be any inequality description, %\footnote{The system may also contain redundant inequalities.} %, but for technical reasons for each $i$,  $P \cap \{ x \mid A_i x=b_i\}$ has to be non-empty.}, 
say with $f$ inequalities. A crucial concept in 
extended formulations is the \emph{slack matrix} $S \in \setR^{f \times v}_{\geq 0}$ which is defined by
$S_{ij} = b_i - A_ix_j$. Moreover, the \emph{non-negative rank} of a matrix is
\[
 \rk_+(S) = \min\{ r \mid \exists U \in \setR_{\geq 0}^{f× r}, V \in \setR_{\geq 0}^{r× v}: S = UV \}.
\]
Recall that if the non-negativity condition is dropped, we recover the usual rank from linear algebra.
The connection between extension complexity and non-negative rank is expressed by the following theorem
(we reprove the statement here to be fully self-contained):
\begin{theorem}[Yannakakis~\cite{ExpressingCOproblemsAsLPs-Yannakakis1991}] \label{thm:Yannakakis}
Let $P$ be a polytope\footnote{For technical reasons we will always assume that the dimension of $P$ is at least $1$.} with vertices $\{ x_1,\ldots,x_v\}$,  inequality
description $P=\{ x \in \setR^n \mid Ax \leq b\}$ and corresponding slack matrix $S$.
Then  $\xc(P) = \rk_+(S)$. 
%Moreover, for any factorization $S = UV$ with $U,V \geq \mathbf{0}$ one can write
% $P = \{ x \in \setR^n \mid \exists y \geq \mathbf{0}: Ax + Uy = b \}$ and for every $x_j \in X$ one has $Ax_j + U\cdot V^j = b$.
\end{theorem}

\begin{proof}
%Let $f$ be the number of inequalities in the system $Ax \leq b$.
%We abbreviate $P = \{ x \in \setR^n \mid Ax \leq b\}$
Let $A_1,\ldots,A_f$ be the rows of matrix $A$. 
%Let $A$ be the matrix consisting of rows $A_1,\ldots,A_f$.
We begin with showing that $r := \rk_+(S) \Rightarrow \xc(P) \leq r$.

So, suppose that we have a non-negative factorization $S = UV$ 
with $U \in \setR^{f \times r}_{\geq 0}$ and $V \in \setR^{r \times v}_{\geq 0}$.
We claim that $Q := \{ (x,y) \in \setR^{n+r} \mid Ax + Uy = b; \; y \geq \bm{0}\}$ is a linear extension and the 
projection $\pi$ with $\pi(x,y) = x$ satisfies that $\pi(Q) = P$; 
in other words, we claim that $P = \{ x \in \setR^n \mid \exists y \in \setR_{\geq 0}^r: Ax + Uy = b\}$. To see this, take a vertex $x_j$ of $P$, then we can choose 
the witness $y := V^j$ and have $(x_j,y) \in Q$ as $A_ix_j + U_iV^j = A_ix_j + S_{ij} = b_i$.
On the other hand, if $x \notin P$, then there is some constraint $i$ 
with $A_ix > b_i$ and no matter what $y \geq \bm{0}$ is chosen, 
we always have $A_ix + U_iy \geq A_ix > b_i$.

For the second part, we have to prove that $r := \xc(P) \Rightarrow \rk_+(S) \leq r$.
Hence, suppose that we have a linear extension $Q = \{ (x,y) \in \setR^{n+k} \mid Bx + Cy \leq d\}$ 
with $r$ inequalities
and a linear projection $\pi$ so that $\pi(Q) = P$. After a linear transformation, 
we may assume that $\pi(x,y) = x$, that means $\pi$ is just the projection on the
$x$-variables. We need to come up with vectors $u_i,v_j \in \setR^r_{\geq 0}$
so that for each constraint $i$ and each vertex $x_j$ one has $\left<u_i,v_j\right> = S_{ij}$.
For each point $x_j$, fix a \emph{lift} $(x_j,y_j) \in Q$ % with $(x_j,y_j) \in Q$. 
and choose $v_j := d - Bx_j - Cy_j \in \setR_{\geq 0}^r$ as the vector of slacks that the lift has
w.r.t. $Q$. 
\begin{figure}
\begin{center}
\psset{xunit=2.0cm,yunit=1.4cm}
\begin{pspicture}(-0.4,-0.1)(2.0,2.8)
%\placeIIID{3}{2}{1}{\pnode(0,0){A}}
\placeIIID{0}{0}{1.5}{\pnode(0,0){q1}} \placeIIID{0}{0}{0}{\pnode(0,0){p1}}
\placeIIID{0}{1}{1.5}{\pnode(0,0){q2}} \placeIIID{0}{1}{0}{\pnode(0,0){p2}}
\placeIIID{1.25}{0.25}{2}{\pnode(0,0){q3}} \placeIIID{1.25}{0.25}{0}{\pnode(0,0){p3}}
\placeIIID{1.25}{1.25}{2}{\pnode(0,0){q4}} \placeIIID{1.25}{1.25}{0}{\pnode(0,0){p4}}
\placeIIID{1}{0}{1}{\pnode(0,0){q5}} \placeIIID{1}{0}{0}{\pnode(0,0){p5}}
\placeIIID{1}{1}{1}{\pnode(0,0){q6}} \placeIIID{1}{1}{0}{\pnode(0,0){p6}}
\pspolygon[linewidth=0.75pt,fillstyle=solid,fillcolor=lightgray](p1)(p2)(p4)(p3)(p5)
%\ncline{q1}{q2} 
%\uncover<3->{
\pspolygon[linewidth=0.75pt,fillstyle=solid,fillcolor=lightgray,opacity=0.5](q1)(q2)(q6)(q5)
\pspolygon[linewidth=0.75pt,fillstyle=solid,fillcolor=lightgray,opacity=0.5](q1)(q3)(q5)
\pspolygon[linewidth=0.75pt,fillstyle=solid,fillcolor=lightgray,opacity=0.5](q3)(q4)(q6)(q5)
\ncline[linestyle=dashed,linewidth=0.5pt]{p1}{q1}
\ncline[linestyle=dashed,linewidth=0.5pt]{p2}{q2}
\ncline[linestyle=dashed,linewidth=0.5pt]{p3}{q3}
\ncline[linestyle=dashed,linewidth=0.5pt]{p4}{q4}
\ncline[linestyle=dashed,linewidth=0.5pt]{p5}{q5}
\ncline[linestyle=dashed,linewidth=0.5pt]{p6}{q6}
\pspolygon[fillstyle=solid,fillcolor=lightgray,opacity=0.7](q1)(q2)(q4)(q3)
\nput[labelsep=30pt]{180}{q6}{$Q$}
%}
%\nput[labelsep=14pt]{180}{q3}{$K$}
% facet i
%\uncover<4->{
\ncline[linewidth=2pt,linecolor=blue]{p1}{p2} \naput[labelsep=0pt,npos=0.6]{$\blue{A_ix+\bm{0}y \leq b_i}$}
\placeIIID{0}{0.5}{0}{\pnode(0,0){f1}}
\placeIIID{-0.5}{0.5}{0}{\pnode(0,0){f2}}
\ncline[linewidth=1.5pt,linecolor=blue,arrowsize=5pt]{->}{f1}{f2}
%}
%\uncover<5->{%----- normal vectors in Q
% upper vector
\placeIIID{0.25}{0.55}{1.6}{\pnode(0,0){qf1head}}% convex combination of (0,0.5,1.5) and (1.25,0.75,2)
\placeIIID{0.1}{0.55}{2.1}{\pnode(0,0){qf1tail}}
\ncline[arrowsize=5pt,linewidth=1pt]{->}{qf1head}{qf1tail}
% lower vector
\placeIIID{0.25}{0.45}{1.3}{\pnode(0,0){qf2head}}% convex combination of (0,0.5,1.5) and (1.25,0.75,2)
\placeIIID{0.19}{0.49}{1.1}{\pnode(0,0){qf2mid}}
\placeIIID{0.1}{0.55}{0.8}{\pnode(0,0){qf2tail}}
\ncline[arrowsize=5pt,linewidth=1pt,linestyle=solid,linecolor=gray]{-}{qf2head}{qf2mid}
\ncline[arrowsize=5pt,linewidth=1pt]{->}{qf2mid}{qf2tail}
%}%
%\uncover<3->{
\psdots[linewidth=1.5pt](q1)(q2)(q3)(q4)(q5)(q6)
\psdots[linewidth=1.5pt](p1)(p2)(p3)(p4)(p5)(p6)
\ncline[linewidth=2pt]{q1}{q2} %\naput[labelsep=0pt]{$\pi^{-1}(G^a)$}
%}
%\nput[labelsep=4pt]{180}{p1}{$c_1c_1^T$}
%\nput[labelsep=4pt]{90}{p2}{$c_2c_2^T$}
%\nput[labelsep=4pt]{180}{q1}{$z^{c_1}$}
%\nput[labelsep=4pt]{90}{q2}{$z^{c_2}$}
\nput[labelsep=4pt]{-90}{p5}{$\textcolor{red!50!black}{x_j}$}\psdots[linewidth=1.5pt,linecolor=red!50!black](p5)
\nput[labelsep=2pt]{-90}{q5}{$\textcolor{red!50!black}{(x_j,y_j)}$}
\psdots[linewidth=1.5pt,linecolor=red!50!black](q5)
\rput[c](0.75,0.25){$P$}
\end{pspicture}
\caption{Visualization of Yannakakis' Theorem.}
\end{center}
\end{figure}
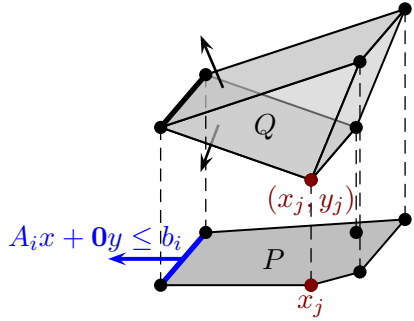
By LP duality we know that each constraint $A_ix + \bm{0}y \leq b_i$ can be derived
as a conic combination of the system $Bx + Cy \leq d$.
In other words, there is a vector $u_i \in \setR^r_{\geq 0}$ so that%\footnote{For technical reasons we assume here that $\dim(P) \geq 1$.}
\[
  u_i^T \begin{pmatrix} B \\ C \\ d \end{pmatrix} = \begin{pmatrix} A_i \\ \bm{0} \\ b_i \end{pmatrix}
\]
Now multiplying gives that %And in fact, multiplying both types of vectors gives
\begin{eqnarray*}
 \left<u_i,v_j\right> &=& \left<u_i,d - Bx_j - Cy_j\right> \\
&=& \underbrace{u_i^Td}_{=b_i} - \underbrace{u_i^TB}_{=A_i}x_j - \underbrace{u_i^TC}_{=\bm{0}} y_j =  b_i - A_ix_j = S_{ij}. 
\end{eqnarray*}
\end{proof}
%In fact, the proof of this Theorem shows that a minimal extended formulation can be obtained
%by factoring $S = UV$ with $U,V \geq \bm{0}$ and writing $P = \{ x \in \setR^n \mid \exists y \geq \bm{0}: Ax + Uy = b\}$.
%Moreover, this theorem is particularly useful to show lower bounds as one can work with \emph{one} concrete object instead of 
%having to argue about \emph{every} extension. 
In particular, Theorem~\ref{thm:Yannakakis} implies that instead of lower bounding the
geometric quantity $\xc(P)$ for a polytope $P$ with slack matrix $S$, it fully suffices to 
find a lower bound for the algebraic quantity $\textrm{rk}_{+}(S)$.
A potential way of lower bounding $\rk_+(S)$ was already pointed
out in the classical paper of Yannakakis and is known as \emph{rectangle covering lower bound}: 
Suppose that $r = \rk_+(S)$ and $S = UV$ with $U,V \geq \bm{0}$. Then
\[
 \textrm{supp}(S) = \{ (i,j) \mid S_{ij} > 0\} = \bigcup_{\ell=1}^r  \left(\{ i \mid U_{i\ell} > 0\} \times \{ j \mid V_{\ell j} > 0 \}\right)
\]
is a covering of the support of $S$ with $r$ rectangles.
In fact,  Fiorini et al.~\cite{TSP-lower-bound-Fiorini-et-al-STOC2012} show that the number of rectangles necessary
for such a covering of the slack-matrix of the correlation polytope is exponential, which in turn lower bounds the
extension complexity. More precisely, this is done by considering a cleverly chosen submatrix of the slack-matrix and then applying Razborov~\cite{Disjointness-Razborov90} as a blackbox.

So, let us discuss the situation for the perfect matching polytope. Since the number of degree constraints
and non-negativity inequalities is polynomial anyway, we consider the part of the slack matrix that is
induced by the odd set inequalities. In other words, we consider the matrix $S$ with
\[
   S_{UM} = |M \cap \delta(U)| - 1  \quad \forall M \subseteq E\textrm{ perfect matching} \quad \forall U \subseteq V: |U| \textrm{ odd}.
\]
The first natural approach would be to check whether the rectangle covering lower bound is superpolynomial. Unfortunately, 
this is not the case, as was already observed in~\cite{ExpressingCOproblemsAsLPs-Yannakakis1991}. %shown in~\cite{CombBoundsOnExtendedFormulations-FioriniKaibelEtAl-DiscMath13}. 
To see this, take any pair $e_1,e_2 \in E$ of non-adjacent edges and choose 
\[
  \mathcal{M}_{e_1,e_2} := \{ M \mid e_1,e_2 \in M\}
\quad \textrm{and} \quad \mathcal{U}_{e_1,e_2} := \{ U \mid e_1,e_2 \in \delta(U) \},
\]
then we obtain $O(n^4)$ many rectangles of the form $\mathcal{U}_{e_1,e_2} \times \mathcal{M}_{e_1,e_2}$.
First of all, we have $S_{UM} \geq |\{ e_1,e_2\}| - 1 \geq 1$ for each $U \in \mathcal{U}_{e_1,e_2}$ and $M \in \mathcal{M}_{e_1,e_2}$, hence
the rectangles contain only entries $(U,M)$ that have positive slack. But every entry $(U,M)$ with $S_{UM}\geq 1$ is 
also contained in at least one such rectangle. To be precise, if $S_{UM} = k$ and $\delta(U) \cap M = \{ e_1,\ldots,e_{k+1}\}$, then
the entry $(U,M)$ lies in ${k+1 \choose 2}$ rectangles. So the approach with the rectangle covering bound does not work.

On the other hand, %even though this is a valid rectangle covering, it 
considering the rectangle covering as a sum of $O(n^4)$ many $0/1$ rank-1 matrices also does not provide a valid non-negative 
factorization of $S$. The reason is that an entry with $S_{UM}=k$ is contained in $\Theta(k^2)$ many rectangles instead of
just $k$ many, thus entries with large slack are \emph{over-covered}. Moreover, we see no way of rescaling the rectangles 
in order to fix the problem. This raises the naive question: 
\begin{quote}
\emph{Maybe every covering of $S$ with polynomially many rectangles must over-cover entries with large slack}?
\end{quote}
Surprisingly, it turns out that the answer is ``\emph{yes}''!

To make this more formal, we will use the \emph{hyperplane separation lower bound} suggested by 
Fiorini~\cite{FioriniCorruptionBound}. %\footnote{The be correct, Fiorini used the term ``corruption lower bound'' (see \url{http://www.dagstuhl.de/mat/index.en.phtml?13082}) as he applied it for the partial disjointness matrix. However, in the general case, I believe the term hyperplane separation lower bound is a better description.}. 
This bound has been known to experts but has not 
appeared explicitly in the literature, hence we state it here in generality and include a proof.
For matrices $S,W \in \setR^{f \times v}$, we will write $\left<S,W\right> = \sum_{i=1}^f \sum_{j=1}^v S_{ij}\cdot W_{ij}$
as their \emph{Frobenius inner product}. Intuitively, the hyperplane separation bound says that if
we can find a linear function $W$ that gives a large value for the slack-matrix $S$, but only 
small values on any rectangle, then the extension complexity is large.  
\begin{lemma}[Hyperplane separation lower bound~\cite{FioriniCorruptionBound}\label{lem:FioriniCorruptionBound}]
Let $S \in \setR_{\geq 0}^{f \times v}$ be the slack-matrix of any polytope $P$ % with extension complexity $r$
and let  $W \in \setR^{f \times v}$ be any matrix. Then
\[
  \xc(P) \geq \frac{\left<W,S\right>}{\| S \|_{\infty} \cdot \alpha}  
\]
with $
 \alpha := \max\{  \left<W,R\right> \mid R \in \{ 0,1\}^{f \times v} \textrm{ rank-1 matrix}\}$.
\end{lemma}
\begin{proof}
First, note that the assumption provides that even for any \emph{fractional} rank-1 matrix $R \in [0,1]^{f \times v}$ one has $\left<W,R\right> \leq \alpha$. 
To see this, take an arbitrary rank-1 matrix $R \in [0,1]^{f \times v}$ and write it as $R = xy^T$ with vectors $x$ and $y$. 
After scaling one 
can assume that $x \in [0,1]^f$ and $y \in [0,1]^v$.
Now suppose that our $R = xy^T$ is an optimal solution to
\[
   \max\Big\{ \sum_{i=1}^f \sum_{j=1}^v W_{ij}\cdot x_iy_j \mid x_i,y_j \in [0,1] \; \forall i \in [f] \; \forall j \in [v] \Big\}
\]
and suppose $R$ is not binary, say because  $x \notin \{ 0,1\}^v$. If we fix $y$, this optimization problem is
linear in $x$ and there is always an $x' \in \{ 0,1\}^v$ that is also optimal. Similarly, also $y$ can be made binary
so that the optimum solution to this LP is a rectangle.  
Geometrically speaking we have just proven that $\textrm{conv}\{ R \in [0,1]^{f \times v}: \textrm{rank}(R) \leq 1\} = \textrm{conv}\{ R \in \{ 0,1\}^{f \times v} : \textrm{rank}(R) \leq 1\}$, even though the set of matrices of rank at most $1$
is not a convex set itself.

By Theorem~\ref{thm:Yannakakis} by have $\xc(P) = \rk_{+}(S)$.  
Now abbreviate $r = \rk_{+}(S)$, then there are $r$ rank-1 matrices $R_1,\ldots,R_r$ with $S = \sum_{i=1}^r R_i$. 
For each $i \in [r]$, we know that $R_i$ is non-negative and hence $\frac{1}{\|R_i\|_{\infty}} R_i$ is a matrix 
with entries in $[0,1]$. We obtain
\[
 \left<W,S\right> = \sum_{i=1}^r \|R_i\|_{\infty} \cdot \underbrace{\left<W,\frac{R_i}{\|R_i\|_{\infty}}\right>}_{\leq \alpha} \leq \alpha \cdot \sum_{i=1}^r \underbrace{\|R_i\|_{\infty}}_{\leq \|S\|_{\infty}} \leq \alpha \cdot r \cdot \|S\|_{\infty}.
\]
Rearranging gives the claim.
\end{proof}

Now, let us go back to the perfect matching polytope and see how we can make use of this bound. 
Let $k \geq 3$ be an odd integer constant that we choose later. We consider 
only complete graphs $G=(V,E)$ that have $|V| = n = 3m(k-3) + 2k$ many vertices, for some odd 
integer\footnote{In other words, we show a lower bound only on $\xc(P_{PM}(n))$ for certain $n$. But note that $P_{PM}(n)$ is a face of $P_{PM}(n')$ for $n \leq n'$ and hence 
$\xc(P_{PM}(n)) \leq \xc(P_{PM}(n'))$. Thus the lower bound on the extension complexity indeed holds for every even $n$.} $m$.
Wherever convenient, we will assume that $n$ and $m$ are large enough, compared to $k$. 

Let $\mathcal{M}_{\textrm{all}} := \{ M \subseteq E \mid M\textrm{ is perfect matching}\}$ be the set of all perfect matchings in $G$.
We fix $t := \frac{m+1}{2}(k-3) + 3$, which is an odd integer, and consider the set
$\mathcal{U}_{\textrm{all}} := \{ U \subseteq V \mid |U| = t \}$ of all $t$-node cuts in $G$. 
Let
\[
  Q_{\ell} := \{ (U,M) \in \mathcal{U}_{\textrm{all}} \times \mathcal{M}_{\textrm{all}} \mid |\delta(U) \cap M| = \ell \}
\]
be the set of pairs of cuts and matchings intersecting in $\ell$ edges and let $\mu_{\ell}$ be the \emph{uniform measure} on $Q_{\ell}$.
In the following, a \emph{rectangle} is of the form $\mathcal{R} = \mathcal{U} \times \mathcal{M}$
with $\mathcal{M} \subseteq \mathcal{M}_{\textrm{all}}$ and
$\mathcal{U} \subseteq \mathcal{U}_{\textrm{all}} $. 
Note that for parity reasons $\mu_{2i}(\mathcal{R}) = 0$ for all  $i \in \setZ_{\geq 0}$. 

Now we want to choose a matrix $W \in \setR^{\mathcal{U}_{\textrm{all}} \times \mathcal{M}_{\textrm{all}}}$ for which the hyperplane separation bound provides an exponential lower bound.
We choose
\[
  W_{U,M} = \begin{cases} 
  -\infty &  |\delta(U) \cap M| = 1 \\
  \frac{1}{|Q_{3}|} & |\delta(U) \cap M| = 3 \\
 -\frac{1}{k-1} \cdot \frac{1}{|Q_k|} & |\delta(U) \cap M| = k \\
  0 & \textrm{otherwise}.
\end{cases}
\]
The intuition is that we reward a rectangle for covering an entry in $Q_3$, punish it for covering
entries in $Q_k$ and completely forbid to cover any entry in $Q_1$. 
First, it is not difficult to see that 
\begin{equation} \label{eq:SWbound}
  \left<W,S\right> = 0  + (3-1)\cdot |Q_3| \cdot \frac{1}{|Q_3|} - (k-1) \cdot |Q_k| \cdot \frac{1}{k-1} \cdot \frac{1}{|Q_k|}  = 1.
\end{equation}
Our hope is that any large rectangle 
$\mathcal{R}$ must over-cover  entries with $|\delta(U) \cap M| = k$ and hence $\left<W,\mathcal{R}\right>$ is small. 
In fact, we can prove
\begin{lemma} \label{lem:RWupperBound}
For any large enough odd constant $k$  ($k := 501$ suffices) and any rectangle $\mathcal{R}$ with $\mathcal{R} = \mathcal{U} \times \mathcal{M}$ with $\mathcal{U} \subseteq \mathcal{U}_{\textrm{all}}$ and $\mathcal{M} \subseteq \mathcal{M}_{\textrm{all}}$
one has $\left<W,\mathcal{R}\right> \leq 2^{-\delta n}$ where $\delta := \delta(k) > 0$ is a constant.
\end{lemma}
The proof of this lemma is the hard part and takes the complete next section. 
From the technical point of view,  our proof is a substantial modification of Razborov's 
original rectangle corruption lemma~\cite{Disjointness-Razborov90}.  

Assuming the bound from Lemma~\ref{lem:RWupperBound} we can then apply Lemma~\ref{lem:FioriniCorruptionBound},
and infer that the perfect matching polytope satisfies 
\[
  \xc(P_{PM}) \geq \frac{\left<W,S\right>}{\|S\|_{\infty} \cdot \max\{ \left<W,\mathcal{R}\right> \mid \mathcal{R}\textrm{ rectangle}\}} \geq \frac{1}{n \cdot 2^{-\delta n}} \geq 2^{\Omega(n)}.
\]
Here we use that $\left<W,S\right> = 1$, $\|S\|_{\infty} \leq n$ and that $\left<W,\mathcal{R}\right> \leq 2^{-\delta n}$ for all rectangles $\mathcal{R}$.

\section{The quadratic measure increase\label{sec:QuadraticMeasureIncrease}}

%We consider a complete graph $G=(V,E)$ with $|V| = 3n(k-3) + 2k$ nodes, where  $k$ is an odd integer and $n$ is even. Note that $|V|$
%is even. We will now show that $\mu_k(\mathcal{R})$ must grow quadratically in $k$ given that  $\mu_1(\mathcal{R}) = 0$.
In this section, we provide the proof of the main technical ingredient, Lemma~\ref{lem:RWupperBound}. 
Formally, we will prove the following statement: 
\begin{lemma} \label{lem:QuadraticGrowOfMeasure}
For each odd $k\geq 3$ and for any rectangle $\mathcal{R}$ with $\mu_1(\mathcal{R}) = 0$, one has % $\mu_k(\mathcal{R}) \geq \frac{k^2}{400} \cdot \mu_3(\mathcal{R}) - 2^{-\delta m}$
$\mu_3(\mathcal{R}) \leq \frac{400}{k^2} \cdot \mu_k(\mathcal{R}) + 2^{-\delta m}$
where $\delta := \delta(k) > 0$.
\end{lemma}
We verify that this indeed implies Lemma~\ref{lem:RWupperBound}. Consider a rectangle $\mathcal{R}$
and assume that $\mu_1(\mathcal{R}) = 0$ since otherwise $\left<W,\mathcal{R}\right> = -\infty$. Then
\[
  \left<W,\mathcal{R}\right> = \mu_3(\mathcal{R}) - \frac{1}{k-1} \mu_k(\mathcal{R}) \stackrel{\textrm{Lem.~\ref{lem:QuadraticGrowOfMeasure}}}{\leq} \underbrace{\Big(\frac{400}{k^2} - \frac{1}{k-1}\Big)}_{\leq 0} \mu_k(\mathcal{R}) + 2^{-\delta m} \leq 2^{-\delta m}
\]
where we choose $k$ as a large enough constant (e.g. $k=501$) and recall that $m$ is linear in $n$. % and assume that $m = \Theta_k(n)$ is large enough.

\subsection{The concept of partitions}

%In the proof, where ever useful we will assume that $n \geq n(k)$ is large enough.
The main trick that Razborov used in his classical paper~\cite{Disjointness-Razborov90} to show a 
relation between certain measures was to
argue that his inequality holds for most \emph{random partitions}, and that the contribution of the
remaining partitions where it does not hold is negligible. 
In fact, we want to use the same rough idea and translate it to the setting of odd cuts and matchings. 
However, our concept of partitions is significantly more involved.
For the remainder of Section~\ref{sec:QuadraticMeasureIncrease}, we fix a rectangle $\mathcal{R} = \mathcal{U} \times \mathcal{M}$ with $\mu_1(\mathcal{R}) = 0$.

A \emph{partition} is a tuple $T = (A=A_1 \dot{\cup} \ldots \dot{\cup} A_m,C,D,B = B_1 \dot{\cup} \ldots \dot{\cup} B_{m})$ with $V = A \dot{\cup} C \dot{\cup} D \dot{\cup} B$ and the following properties:
\begin{itemize}
%\item  $F \subseteq E$ is a  \emph{partial matching} with $|F| = n(k-3)$ edges that is partitioned into blocks $F = F_1 \dot{\cup} \ldots \dot{\cup} F_n$ with $|F_i| = k-3$ edges each
\item  $A \subseteq V$ is a set of $|A| = m(k-3)$ nodes that is partitioned into blocks $A = A_1 \dot{\cup} \ldots \dot{\cup} A_m$
with $|A_i| = k-3$ nodes each.
\item $C \subseteq V$ is a set of $k$ nodes.
\item $D \subseteq V$ is a set of $k$ nodes.
\item $B = B_1 \dot{\cup} \ldots \dot{\cup} B_{m}$ with $B \subseteq V$ is a partition of the remaining nodes 
so that $|B_i| = 2(k-3)$.
% $|B| = 4n$ nodes, so that there is a unique index $i^* \in [\tilde{m}]$
%so that 
%\[
%  |B_i| = \begin{cases} 2(k-3) & \textrm{if } i \neq i^* \\
% k & \textrm{if } i=i^* \end{cases}
%\]
%The block $B_{i^*}$ is also called \emph{special}.
\end{itemize}
%Overall the partition splits the node set into
%\[
%  V = \underbrace{F_1}_{\textrm{size }k-3} \dot{\cup} \ldots \dot{\cup}  \underbrace{F_n}_{\textrm{size } k-3} %
% \dot{\cup} \underbrace{C}_{\textrm{size }k} \dot{\cup} \underbrace{D}_{\textrm{size }k} \dot{\cup}  \underbrace{B_1}_{\textrm{size }2(k-3)} \dot{\c%up} \ldots \dot{\cup}  \underbrace{B_n}_{\textrm{size } 2(k-3)}
%\]
%Intuitively, the nodes in the special block $D$ are used to match some of the nodes in $C$.
Here, the symbol ``$\dot{\bigcup}$'' indicates a union of disjoint set.
 For a node-set $U$, let $E(U) := \{ (u,v) \in E \mid u,v \in U\}$ be the edges lying inside of $U$. We abbreviate 
$E(T) := \bigcup_{i=1}^m E(A_i) \cup E(C \cup D) \cup \bigcup_{i=1}^m E(B_i)$ as the edges associated with the partition $T$, see Figure~\ref{fig:EdgesET}.
\begin{figure}
\begin{center}
\psset{unit=0.7cm}
\begin{pspicture}(0,-0.5)(14,7.5)
\drawRect{fillcolor=red!30!white,fillstyle=solid,linearc=0.1}{0}{0}{6}{7} % F
\drawRect{fillcolor=green!30!gray,fillstyle=solid,linearc=0.1}{6.2}{0.0}{0.6}{7} % T_0=C
\drawRect{fillcolor=blue!30!white,fillstyle=solid,linearc=0.1}{8.0}{0.0}{6}{7} % whole B
\drawRect{fillcolor=orange!50!white,fillstyle=solid,linearc=0.1}{7.2}{0.0}{0.6}{7} % D=B_i^*
\drawRect{fillcolor=blue!50!white,fillstyle=solid,linearc=0.1}{8.2}{0.2}{1.6}{6.6} %
\drawRect{fillcolor=blue!50!white,fillstyle=solid,linearc=0.1}{10.2}{0.2}{1.6}{6.6} %
\drawRect{fillcolor=blue!50!white,fillstyle=solid,linearc=0.1}{12.2}{0.2}{1.6}{6.6} %
\rput[c](2.5,7.5){\Large{\darkgray{$A$}}}
\rput[c](6.5,7.5){\Large{\darkgray{$C$}}}
\rput[c](7.5,7.5){\Large{\darkgray{$D$}}}
\rput[c](11.5,7.5){\Large{\darkgray{$B$}}}
\rput[c](9,5.5){\Large{\darkgray{$B_1$}}}
\rput[c](11,5.5){\Large{\darkgray{$\ldots$}}}
\rput[c](13,5.5){\Large{\darkgray{$B_{m}$}}}
%\rput[c](6.5,7.3){\Large{\darkgray{$U$}}}
% nodes and edges in F
\drawRect{fillstyle=solid,fillcolor=red!50!white,linearc=0.1}{0.2}{0.7}{1.6}{1.6}
\drawRect{fillstyle=solid,fillcolor=red!50!white,linearc=0.1}{0.2}{2.7}{1.6}{1.6}
\drawRect{fillstyle=solid,fillcolor=red!50!white,linearc=0.1}{0.2}{4.7}{1.6}{1.6}
\drawRect{fillstyle=solid,fillcolor=red!50!white,linearc=0.1}{2.2}{0.7}{1.6}{1.6}
\drawRect{fillstyle=solid,fillcolor=red!50!white,linearc=0.1}{2.2}{2.7}{1.6}{1.6}
\drawRect{fillstyle=solid,fillcolor=red!50!white,linearc=0.1}{2.2}{4.7}{1.6}{1.6}
\drawRect{fillstyle=solid,fillcolor=red!50!white,linearc=0.1}{4.2}{0.7}{1.6}{1.6}
\drawRect{fillstyle=solid,fillcolor=red!50!white,linearc=0.1}{4.2}{2.7}{1.6}{1.6}
\drawRect{fillstyle=solid,fillcolor=red!50!white,linearc=0.1}{4.2}{4.7}{1.6}{1.6}
%\pspolygon[fillstyle=vlines,hatchcolor=gray,linecolor=darkgray,linewidth=1.5pt,linearc=0.1](0.1,6.6)(6,6.6)(6,6.9)(7,6.9)(7,4.1)(6,4.1)(6,4.5)(4,4.5)(4,2.5)(2,2.5)(2,4.5)(0.1,4.5) % U
\multido{\N=0.5+2.0}{3}{ % F
  \multido{\n=1.0+2.0}{3}{%
     \rput[c](\N,\n){%
        \cnode*(0,0){2.5pt}{A}%
        \cnode*(1,0){2.5pt}{B}%
        \cnode*(0,1){2.5pt}{C}%
        \cnode*(1,1){2.5pt}{D}%
        \ncline[linewidth=0.5pt,linecolor=black]{A}{B}%
       \ncline[linewidth=0.5pt,linecolor=black]{A}{C}%
       \ncline[linewidth=0.5pt,linecolor=black]{A}{D}%
       \ncline[linewidth=0.5pt,linecolor=black]{B}{C}%
       \ncline[linewidth=0.5pt,linecolor=black]{B}{D}%
       \ncline[linewidth=0.5pt,linecolor=black]{C}{D}%
     }
  }
}
\rput[c](1,5.5){\Large{\darkgray{$A_1$}}}
\rput[c](3,5.5){\Large{\darkgray{$A_2$}}}
\rput[c](5,5.5){\Large{\darkgray{$A_3$}}}
\rput[c](3,1.5){\Large{\darkgray{$\ldots$}}}
\rput[c](5,1.5){\Large{\darkgray{$A_m$}}}
\psset{linewidth=0.5pt,linecolor=black} % B
\multido{\N=8.5+2.0}{3}{
 \rput[c](\N,2){%
  \cnode*(0,3){2.5pt}{A1}  \cnode*(1,3){2.5pt}{B1}%
  \cnode*(0,2){2.5pt}{A2}  \cnode*(1,2){2.5pt}{B2}%
  \cnode*(0,1){2.5pt}{A3}  \cnode*(1,1){2.5pt}{B3}%
  \cnode*(0,0){2.5pt}{A4}  \cnode*(1,0){2.5pt}{B4}%
  \ncline{A1}{B1}%
  \ncline{A1}{B2}%
  \ncline{A1}{B3}%
  \ncline{A1}{B4}%
  \ncline{A2}{B1}%
  \ncline{A2}{B2}%
  \ncline{A2}{B3}%
  \ncline{A2}{B4}%
  \ncline{A3}{B1}%
  \ncline{A3}{B2}%
  \ncline{A3}{B3}%
  \ncline{A3}{B4}%
  \ncline{A4}{B1}%
  \ncline{A4}{B2}%
  \ncline{A4}{B3}%
  \ncline{A4}{B4}%
  \ncline{A1}{A4}%
  \ncline{B1}{B4}%
 }
}
\multido{\N=0.5+1.0}{7}{%
  \multido{\n=0.5+1.0}{7}{%
     \cnode*(6.5,\N){2.5pt}{A}%
     \cnode*(7.5,\n){2.5pt}{B}%
     \ncline{A}{B}%
  }
}
\psline(6.5,0.5)(6.5,6.5)
\psline(7.5,0.5)(7.5,6.5)
\end{pspicture}
\caption{Visualization of a partition $T$ with all edges $E(T)$.\label{fig:EdgesET}}
\end{center}
\end{figure}
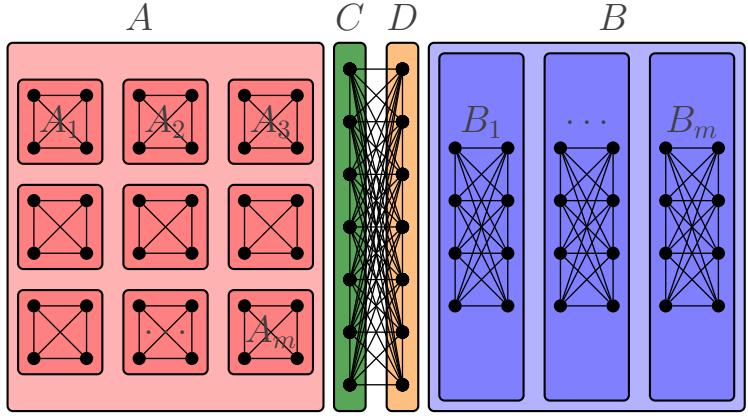

\begin{figure}
\begin{center}
\psset{unit=0.7cm}
\begin{pspicture}(0,-3.5)(14,7.5)
\drawRect{fillcolor=red!30!white,fillstyle=solid,linearc=0.1}{0}{0}{6}{7} % F
\drawRect{fillcolor=green!30!gray,fillstyle=solid,linearc=0.1}{6.2}{0.0}{0.6}{7} % T_0=C
\drawRect{fillcolor=blue!30!white,fillstyle=solid,linearc=0.1}{8.0}{0.0}{6}{7} % whole B
\drawRect{fillcolor=orange!50!white,fillstyle=solid,linearc=0.1}{7.2}{0.0}{0.6}{7} % D=B_i^*
\drawRect{fillcolor=blue!50!white,fillstyle=solid,linearc=0.1}{8.2}{0.2}{1.6}{6.6} %
\drawRect{fillcolor=blue!50!white,fillstyle=solid,linearc=0.1}{10.2}{0.2}{1.6}{6.6} %
\drawRect{fillcolor=blue!50!white,fillstyle=solid,linearc=0.1}{12.2}{0.2}{1.6}{6.6} %
\rput[c](6.5,2.0){\Large{\darkgray{$C$}}}
\rput[c](7.5,4.0){\Large{\darkgray{$D$}}}
\rput[c](9,3.5){\Large{\darkgray{$B_1$}}}
\rput[c](11,3.5){\Large{\darkgray{$\ldots$}}}
\rput[c](13,3.5){\Large{\darkgray{$B_{m}$}}}
\rput[c](6.5,7.3){\Large{\darkgray{$U$}}}
% nodes and edges in F
\drawRect{fillstyle=solid,fillcolor=red!50!white,linearc=0.1}{0.2}{0.7}{1.6}{1.6}
\drawRect{fillstyle=solid,fillcolor=red!50!white,linearc=0.1}{0.2}{2.7}{1.6}{1.6}
\drawRect{fillstyle=solid,fillcolor=red!50!white,linearc=0.1}{0.2}{4.7}{1.6}{1.6}
\drawRect{fillstyle=solid,fillcolor=red!50!white,linearc=0.1}{2.2}{0.7}{1.6}{1.6}
\drawRect{fillstyle=solid,fillcolor=red!50!white,linearc=0.1}{2.2}{2.7}{1.6}{1.6}
\drawRect{fillstyle=solid,fillcolor=red!50!white,linearc=0.1}{2.2}{4.7}{1.6}{1.6}
\drawRect{fillstyle=solid,fillcolor=red!50!white,linearc=0.1}{4.2}{0.7}{1.6}{1.6}
\drawRect{fillstyle=solid,fillcolor=red!50!white,linearc=0.1}{4.2}{2.7}{1.6}{1.6}
\drawRect{fillstyle=solid,fillcolor=red!50!white,linearc=0.1}{4.2}{4.7}{1.6}{1.6}
\pspolygon[fillstyle=vlines,hatchcolor=gray,linecolor=darkgray,linewidth=1.5pt,linearc=0.1](0.1,6.6)(6,6.6)(6,6.9)(7,6.9)(7,4.1)(6,4.1)(6,4.5)(4,4.5)(4,2.5)(2,2.5)(2,4.5)(0.1,4.5) % U
\rput[c](1,5.5){\Large{\darkgray{$A_1$}}}
\rput[c](3,5.5){\Large{\darkgray{$A_2$}}}
\rput[c](5,5.5){\Large{\darkgray{$A_3$}}}
\rput[c](3,1.5){\Large{\darkgray{$\ldots$}}}
\rput[c](5,1.5){\Large{\darkgray{$A_m$}}}
% column 1
\cnode*(0.5,1.0){2.5pt}{F11} \cnode*(1.5,1.0){2.5pt}{F12} 
\cnode*(0.5,2.0){2.5pt}{F21} \cnode*(1.5,2.0){2.5pt}{F22} \ncline[linewidth=1pt]{F11}{F21} \ncline[linewidth=1pt]{F12}{F22}
\cnode*(0.5,3.0){2.5pt}{F31} \cnode*(1.5,3.0){2.5pt}{F32} \ncline[linewidth=1pt]{F31}{F32}
\cnode*(0.5,4.0){2.5pt}{F41} \cnode*(1.5,4.0){2.5pt}{F42} \ncline[linewidth=1pt]{F41}{F42}
\cnode*(0.5,5.0){2.5pt}{F51} \cnode*(1.5,5.0){2.5pt}{F52} \ncline[linewidth=1pt]{F51}{F52}
\cnode*(0.5,6.0){2.5pt}{F61} \cnode*(1.5,6.0){2.5pt}{F62} \ncline[linewidth=1pt]{F61}{F62}
% column 2
\cnode*(2.5,1.0){2.5pt}{F71} \cnode*(3.5,1.0){2.5pt}{F72} \ncline[linewidth=1pt]{F71}{F72}
\cnode*(2.5,2.0){2.5pt}{F81} \cnode*(3.5,2.0){2.5pt}{F82} \ncline[linewidth=1pt]{F81}{F82}
\cnode*(2.5,3.0){2.5pt}{F91} \cnode*(3.5,3.0){2.5pt}{F92} 
\cnode*(2.5,4.0){2.5pt}{F101} \cnode*(3.5,4.0){2.5pt}{F102} \ncline[linewidth=1pt]{F91}{F102} \ncline[linewidth=1pt]{F101}{F92}
\cnode*(2.5,5.0){2.5pt}{F111} \cnode*(3.5,5.0){2.5pt}{F112} \ncline[linewidth=1pt]{F111}{F112}
\cnode*(2.5,6.0){2.5pt}{F121} \cnode*(3.5,6.0){2.5pt}{F122} \ncline[linewidth=1pt]{F121}{F122}
% column 3
\cnode*(4.5,1.0){2.5pt}{F131} \cnode*(5.5,1.0){2.5pt}{F132} \ncline[linewidth=1pt]{F131}{F132}
\cnode*(4.5,2.0){2.5pt}{F141} \cnode*(5.5,2.0){2.5pt}{F142} \ncline[linewidth=1pt]{F141}{F142}
\cnode*(4.5,3.0){2.5pt}{F151} \cnode*(5.5,3.0){2.5pt}{F152} \ncline[linewidth=1pt]{F151}{F152}
\cnode*(4.5,4.0){2.5pt}{F161} \cnode*(5.5,4.0){2.5pt}{F162} \ncline[linewidth=1pt]{F161}{F162}
\cnode*(4.5,5.0){2.5pt}{F171} \cnode*(5.5,5.0){2.5pt}{F172} \ncline[linewidth=1pt]{F171}{F172}
\cnode*(4.5,6.0){2.5pt}{F181} \cnode*(5.5,6.0){2.5pt}{F182} \ncline[linewidth=1pt]{F181}{F182}
% core C
\cnode*(6.5,6.5){2.5pt}{C1} 
\cnode*(6.5,5.5){2.5pt}{C2} 
\cnode*(6.5,4.5){2.5pt}{C3} 
\cnode*(6.5,3.5){2.5pt}{C4} 
\cnode*(6.5,2.5){2.5pt}{C5} 
\cnode*(6.5,1.5){2.5pt}{C6} 
\cnode*(6.5,0.5){2.5pt}{C7}
% B*
\cnode*(7.5,6.5){2.5pt}{BS1}
\cnode*(7.5,5.5){2.5pt}{BS2}
\cnode*(7.5,4.5){2.5pt}{BS3}
\cnode*(7.5,3.5){2.5pt}{BS4}
\cnode*(7.5,2.5){2.5pt}{BS5}
\cnode*(7.5,1.5){2.5pt}{BS6}
\cnode*(7.5,0.5){2.5pt}{BS7}
\ncline[linewidth=2pt]{C1}{BS1}
\ncline[linewidth=2pt]{C2}{BS2} \naput[labelsep=2pt]{$H$}
\ncline[linewidth=2pt]{C3}{BS3}
\ncline[linewidth=1pt]{C4}{C5}
\ncline[linewidth=1pt]{C6}{C7}
\ncarc[linewidth=1pt]{BS4}{BS7} 
\ncline[linewidth=1pt]{BS5}{BS6}
% B
\cnode*(8.5,5){2.5pt}{B11} \cnode*(8.5,4){2.5pt}{B12} \cnode*(8.5,3){2.5pt}{B13} \cnode*(8.5,2){2.5pt}{B14}
\cnode*(9.5,5){2.5pt}{B21} \cnode*(9.5,4){2.5pt}{B22} \cnode*(9.5,3){2.5pt}{B23} \cnode*(9.5,2){2.5pt}{B24}
\cnode*(10.5,5){2.5pt}{B31} \cnode*(10.5,4){2.5pt}{B32} \cnode*(10.5,3){2.5pt}{B33} \cnode*(10.5,2){2.5pt}{B34}
\cnode*(11.5,5){2.5pt}{B41} \cnode*(11.5,4){2.5pt}{B42} \cnode*(11.5,3){2.5pt}{B43} \cnode*(11.5,2){2.5pt}{B44}
\cnode*(12.5,5){2.5pt}{B51} \cnode*(12.5,4){2.5pt}{B52} \cnode*(12.5,3){2.5pt}{B53} \cnode*(12.5,2){2.5pt}{B54}
\cnode*(13.5,5){2.5pt}{B61} \cnode*(13.5,4){2.5pt}{B62} \cnode*(13.5,3){2.5pt}{B63} \cnode*(13.5,2){2.5pt}{B64}
% B1
\ncline[linewidth=1pt]{B11}{B21}
\ncline[linewidth=1pt]{B12}{B22}
\ncline[linewidth=1pt]{B13}{B24}
\ncline[linewidth=1pt]{B14}{B23}
% B2
\ncline[linewidth=1pt]{B31}{B42}
\ncline[linewidth=1pt]{B32}{B41}
\ncline[linewidth=1pt]{B33}{B34}
\ncline[linewidth=1pt]{B43}{B44}
% B3
\ncline[linewidth=1pt]{B51}{B61}
\ncline[linewidth=1pt]{B52}{B62}
\ncline[linewidth=1pt]{B53}{B64}
\ncline[linewidth=1pt]{B54}{B63}
\psbrace[ref=1C,rot=90,nodesepB=5pt,braceWidthInner=4pt,braceWidthOuter=4pt](0,0)(2,0){$k-3$}
\psbrace[ref=1C,rot=90,nodesepB=5pt,braceWidthInner=4pt,braceWidthOuter=4pt](6,0)(7,0){$k$}
\psbrace[ref=1C,rot=90,nodesepB=5pt,braceWidthInner=4pt,braceWidthOuter=4pt](7,0)(8,0){$k$}
%\psbrace[ref=1C,rot=90,nodesepB=5pt,braceWidthInner=4pt,braceWidthOuter=4pt](8,0)(10,0){$2(k-3)$}
%\psbrace[ref=1C,rot=90,nodesepB=5pt,braceWidthInner=4pt,braceWidthOuter=4pt](10,0)(12,0){$2(k-3)$}
\psbrace[ref=1C,rot=90,nodesepB=5pt,braceWidthInner=4pt,braceWidthOuter=4pt](12,0)(14,0){$2(k-3)$}
\psbrace[ref=1C,rot=90,nodesepB=18pt](0,-1)(6,-1){$\begin{array}{c} m\textrm{ blocks, each }k-3\textrm{ nodes} \\ |A| = m(k-3) \textrm{ nodes} \end{array}$}
\psbrace[ref=1C,rot=90,nodesepB=18pt](8,-1)(14,-1){$\begin{array}{c} m\textrm{ blocks, each }2(k-3)\textrm{ nodes} \\ |B| = m\cdot 2(k-3) \textrm{ nodes} \end{array}$}
\end{pspicture}
\caption{Visualization of a partition $T$ together with one matching $M\in \mathcal{M}(T)$ and one cut $U \in \mathcal{U}(T)$.\label{fig:ExampleCutAndMatching}}
\end{center}
\end{figure}
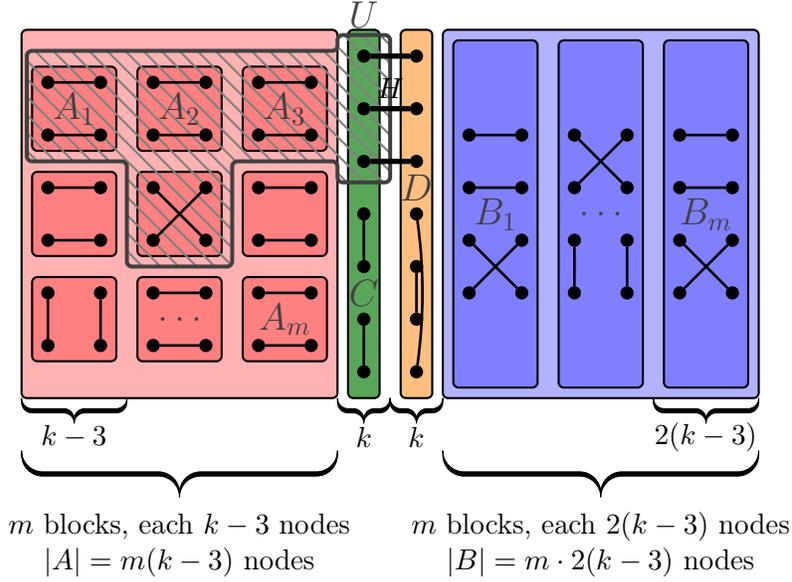

We say that 
\[
\mathcal{M}_{\textrm{all}}(T) := \{ M \in \mathcal{M}_{\textrm{all}} \mid M \subseteq E(T) \}
\]
are all perfect matchings that respect the partition $T$.
%\[
%  \mathcal{M}(T) := \{ M \in \mathcal{M} \mid F \subseteq M; \; \delta(B_i) \cap M = \emptyset \; \forall i\in [n]  \}
%\] 
%respect the partition $T$. 
In other words, the matchings in $\mathcal{M}_{\textrm{all}}(T)$ have only edges running inside $A_i$ or $B_i$ or inside $C \cup D$.
Similarly, we say that 
\[
  \mathcal{U}(T) := \{ U \in \mathcal{U}_{\textrm{all}} \mid U \subseteq A \cup C\textrm{ with } |U \cap A_i| \in \{ 0,|A_i|\} \; \forall i \in [m] \}
\]
are all the $t$-node cuts that respect the partition (see Figure~\ref{fig:ExampleCutAndMatching}). 
In other words, those cuts are fully contained in $A \cup C$ and for each $i$, they contain either all or none of the nodes $A_i$. %In particular, 
%the onla cut $U \in \mathcal{U}(U)$ can only cut edges 
%those cuts $U$
%are never going to cut any edge inside $A_i$.
Moreover, let 
\[
\mathcal{M}(T) := \mathcal{M} \cap \mathcal{M}_{\textrm{all}}(T) = \{ M \in \mathcal{M} \mid M \subseteq E(T) \}
\]
and 
\[
\mathcal{U}(T) := \mathcal{U} \cap \mathcal{U}_{\textrm{all}}(T) =\{ U \in \mathcal{U} \mid U \subseteq A \cup C; \; |U \cap A_i| \in \{ 0,|A_i|\} \; \forall i \in [m]\}
\] 
be the subsets containing all 
 matchings and cuts from our rectangle $\mathcal{R}$ that respect the
partition. 
The advantage of such partitions $T$ is that if we take a matching $M \in \mathcal{M}_{\textrm{all}}(T)$
and a cut $U \in \mathcal{U}_{\textrm{all}}(T)$, then the intersection $\delta(U) \cap M$ can only contain edges in $E(C \cup D)$
(in fact, it contains an odd number between $1$ and $k$ edges). 
%Note that for any pair $(U,M) \in \mathcal{U}(T) \times \mathcal{M}(T)$, any
%intersection can only take place in the core $C$ and moreover $M$ has between $1$ and $k$ edges running between $C$ and $D$.

\subsection{Generating the distributions  $\mu_3$ and $\mu_k$}

The key trick is that the measures $\mu_3(\mathcal{R})$ and $\mu_k(\mathcal{R})$ can be nicely 
compared for the rectangles $\mathcal{U}(T) \times \mathcal{M}(T)$ that are induced by each partition $T$. 
Hence, we consider an
alternative way to generate uniform members of $Q_3$ and $Q_k$.  To fix some notation, we say that $H$ is an \emph{$\ell$-matching}
if $H$ is a matching with exactly $\ell$ edges. The nodes incident to edges $H$ are denoted by $V(H)$.

For a matching $H \subseteq E(C \cup D)$, we define $p_{\mathcal{M},T}(H) := \Pr_{M \sim \mathcal{M}_{\textrm{all}}(T)}[M \in \mathcal{M} \mid H \subseteq M]$ as the chance that a random extension of $H$ to a perfect matching respecting the partition lies in the 
rectangle. Here, $M \sim \mathcal{M}_{\textrm{all}}(T)$ indicates that $M$ is a \emph{uniformly} drawn matching from $\mathcal{M}_{\textrm{all}}(T)$.
If $H \subseteq C \times D$, then we also define
 $p_{\mathcal{M},T}^{\textrm{ex}}(H) := \Pr_{M \sim \mathcal{M}_{\textrm{all}}(T)}[M \in \mathcal{M} \mid M \cap (C \times D) = H]$ as the probability given that $H$ is the {\bf ex}clusive set of 
edges that runs between $C$ and $D$. 
For $c \subseteq C$, let $p_{\mathcal{U},T}(c) := \Pr_{U \sim \mathcal{U}_{\textrm{all}}(T)}[U \in \mathcal{U} \mid c \subseteq U]$
be the chance that a random extension of $c$ to a cut $U$ respecting the partition $T$ lies in the rectangle. 
Again, we also define $p_{\mathcal{U},T}^{\textrm{ex}}(c) := \Pr_{U \sim \mathcal{U}_{\textrm{all}}(T)}[U \in \mathcal{U} \mid U \cap C = c]$
as the chance, given that $c$ are the only nodes in $C$. 
By a slight abuse of notation, we denote $p_{\mathcal{U},T}^{\textrm{ex}}(H) := p_{\mathcal{U},T}^{\textrm{ex}}(V(H) \cap C)$ for
a matching $H \subseteq C \times D$. Recall that all cuts $U \in \mathcal{U}_{\textrm{all}}$ have size $|U| = t$;  
this implies that $|U| - 3$ is a multiple of $k-3$, and hence $p_{\mathcal{U},T}(c) > 0$ only if $|c| \in \{ 3,k\}$.
%that a random extension that {\bf ex}clusively has $H$ going between $C$
%and $D$, lies in the rectangle. For any matching $H \subseteq E(C \cup D)$ (that also may contain edges running
%inside $C$ and $D$), we have .
%Again only for $H \subseteq C \times D$, we define
% $p_{\mathcal{U},T}(H) := \Pr_{U \in \mathcal{U}_{\textrm{all}}(T)}[U \in \mathcal{U} \mid U \cap C = V(H) \cap C]$ %. Intuitively, $p_{\mathcal{M},T}(H)$ is the chance that a random extension of $H$ to a perfect matching lies in our rectangle and $p_{\mathcal{U},T}(H)$ is the 
%as the chance that a random extension of the nodes in $V(H) \cap C$
%to a cut $U$ lies in our rectangle. 

For the sake of a clearer notation, in the following we will use the symbol $H$ always for a 3-matching $H \subseteq C \times D$
and the symbol $F$ will always be used for a $k$-matching that we take either as $F \subseteq E(C \cup D)$ or as $F \subseteq C \times D$.
In the remainder of this paper, whenever 
we write $\E_T[\ldots]$, then $T$ is a uniform random partition and if we write $\E_{|H| = 3}[\ldots]$, then
$H$ is a uniformly picked 3-matching in the complete bipartite graph between $C$ and $D$ (always assuming that
the partition $T$ has been selected before).
The concept of partitions can be used to generate uniform entries from $Q_3$ and $Q_k$  
in the following way: 
\begin{itemize}
\item \emph{Generating a uniform random entry $(U,M) \in Q_3$:} Pick a random partition $T$. 
Pick a random 3-matching $H \subseteq C \times D$. Then randomly extend $H$ to a matching $M \in \mathcal{M}_{\textrm{all}}(T)$ with $M \cap (C \times D) = H$
and to a cut $U \in \mathcal{U}_{\textrm{all}}(T)$ with $\delta(U) \cap (C \times D) = H$. %Note that this process is
%fully symmetric with respect to node-permutation
Hence\footnote{We should also argue that $(U,M)$ is indeed a 
uniform random element from $Q_3$. First, note that by construction always 
$(U,M) \in Q_3$. Secondly, $Q_3$ is symmetric under permuting nodes. Also the produced pair $(U,M)$ is symmetric under permuting nodes. That gives the claim.}
\[
  \mu_3(\mathcal{R}) = \E_T\Big[ \E_{|H| = 3}[ p_{\mathcal{M},T}^{\textrm{ex}}(H) \cdot p_{\mathcal{U},T}^{\textrm{ex}}(H)]\Big]
\]
See again Figure~\ref{fig:ExampleCutAndMatching}. 
%Note that the symmetry of $Q_3$ 
\item \emph{Generating a uniform random entry $(U,M) \in Q_k$:} Pick a random partition $T$. Pick a random $k$-matching $F$
in the bipartite graph $C \times D$. 
Then randomly extend $F$ to a matching $M \in \mathcal{M}_{\textrm{all}}(T)$ with $M \cap (C \times D) = F$ and to a 
cut $U \in \mathcal{U}_{\textrm{all}}(T)$ with $\delta(U) \cap (C \times D) = F$. Hence
\[
 \mu_k(\mathcal{R}) = \E_T\Big[ \E_{|F| = k} [ p_{\mathcal{M},T}^{\textrm{ex}}(F) \cdot p_{\mathcal{U},T}^{\textrm{ex}}(F) ] \Big]
\]
\end{itemize}
Note that for a $k$-matching $F \subseteq C \times D$, we anyway have $p_{\mathcal{M},T}(F) = p_{\mathcal{M},T}^{\textrm{ex}}(F)$ and $p_{\mathcal{U},T}(F) = p_{\mathcal{U},T}^{\textrm{ex}}(F)$.

\subsection{The notion of good pairs}

An important definition is the one of \emph{good} pairs, which are those pairs $(T,H)$ for 
which we can easily show that their contribution to $\mu_3(\mathcal{R})$ is only a $O(\frac{1}{k^2})$-fraction of the contribution
to $\mu_k(\mathcal{R})$. In the following, $\varepsilon>0$ denotes a small enough constant that we determine later
(in fact $\varepsilon=\frac{1}{8}$ will suffice). 
%Intuitively we will call a pair $(T,H)$ good if the matchings in $\mathcal{M}(T)$ containing $H$
%behave like a uniform distribution on the $H$-exposed nodes in $C \cup D$ and additionally the cuts
%in $\mathcal{U}(T)$ have a roughly equal probability of containing all or none of $C \backslash V(H)$.
%More formally, we use the following definition:
\begin{definition}[\emph{$\mathcal{M}$-good}]
Let $T$ be a partition and  $H \subseteq C \times D$ be a 3-matching. 
The pair $(T,H)$ is called  $\mathcal{M}$-\emph{good} if % $p_{\mathcal{M},T}(H) = (1 \pm \varepsilon) p_{\mathcal{M},T}(F) > 0$ 
$0<\frac{1}{1+\varepsilon} p_{\mathcal{M},T}(H) \leq p_{\mathcal{M},T}(F) \leq (1+\varepsilon)p_{\mathcal{M},T}(H)$
for all $k$-matchings $F$ with $H \subseteq F \subseteq E(C \cup D)$. 
%The pair $(T,H)$ is called \emph{$M$-small} if $p_{\mathcal{M},T}(H) \leq 2^{-\delta n}$.
%If the pair $(T,H)$ is neither $\mathcal{M}$-good, nor $M$-small, then we call it \emph{$M$-bad}. 
\end{definition}
While we prefer this as a formal definition, there is a more intuitive one: imagine we draw a random 
matching $M$ from $\{ M' \in \mathcal{M}(T) \mid H \subseteq M' \}$. 
If $(T,H)$ is $\mathcal{M}$-good, then the induced random matching
$M \cap E((C \cup D) \backslash V(H))$ is \emph{$\varepsilon$-close}\footnote{We will not further use the notion of $\varepsilon$-closeness. However, 
a proper definition would be to call two distributions $\varepsilon$-close if all atomic events have the same probability up to a multiplicative $1+\varepsilon$ factor.} to the distribution of a \emph{uniform} random matching on the node-set $(C \cup D) \backslash V(H)$.
The arguments based on conditional probability can be seen in Cor.~\ref{cor:Pseudo1} in Section~\ref{sec:PseudoRandomBehavior}.
In particular this will imply that for a good pair $(T,H)$, all edges in $E((C \cup D) \backslash V(H))$ are contained in at least one matching in $\mathcal{M}(T)$.
We have only used $p_{\mathcal{M},T}$ in the definition of $\mathcal{M}$-goodness and not $p_{\mathcal{M},T}^{\textrm{ex}}$. But this is easy to extend:
\begin{lemma} \label{lem:pEx-vs-p}
Let $T$ be a partition and $H \subseteq C \times D$ be a 3-matching.  
If $(T,H)$ is $\mathcal{M}$-good, then % $p_{\mathcal{M},T}^{\textrm{ex}}(H) = (1 \pm 2\varepsilon) p_{\mathcal{M},T}(H)$.
$\frac{1}{1+\varepsilon} p_{\mathcal{M},T}(H) \leq p_{\mathcal{M},T}^{\textrm{ex}}(H) \leq (1+\varepsilon) p_{\mathcal{M},T}(H)$.
% $p_{\mathcal{M},T}^{\textrm{ex}}(H) = (1 \pm 3\varepsilon) p_{\mathcal{M},T}(F) > 0$
%for all $k$-matchings $F$ with $H \subseteq F \subseteq E(C \cup D)$.
\end{lemma}
\begin{proof}
Fix $H$ and take a uniform random $k$-matching $F \sim \{ F' \subseteq E(C \cup D) \mid |F'| = k\textrm{ and }(C \times D) \cap F' = H\}$.
Then we can write 
$p_{\mathcal{M},T}^{\textrm{ex}}(H) = \E_{F}[ p_{\mathcal{M},T}(F) ] \leq (1+\varepsilon) \cdot p_{\mathcal{M},T}(H)$ 
(analogously for the lower bound). 
%and similarl
% $p_{\mathcal{M},T}^{\textrm{ex}}(H) = \E_F[p_{\mathcal{M},T}(F)] \geq \frac{1}{1+\varepsilon}p_{\mathcal{M},T}(H) \geq (1-2\varepsilon)p_{\mathcal{M},T}(H)$
 % \leq (1+\varepsilon)^2 \cdot p_{\mathcal{M},T}(F) \leq (1+3\varpepsilon)$,
%where we use that $(T,H)$ is $\mathcal{M}$-good. The lower bound holds analogously. % and that $\varepsilon \leq \frac{1}{8}$.
%for $0<\varepsilon\leq \frac{1}{2}$.
\end{proof}

\begin{definition}[\emph{$\mathcal{U}$-good}]
Let $T$ be a partition and  $H \subseteq C \times D$ be a 3-matching. 
The pair  $(T,H)$ is called $\mathcal{U}$-\emph{good} if  % $p_{\mathcal{U},T}^{\textrm{ex}}(H) = (1 \pm \varepsilon)p_{\mathcal{U},T}^{\textrm{ex}}(C)$.
$0<\frac{1}{1+\varepsilon} p_{\mathcal{U},T}^{\textrm{ex}}(H) \leq p_{\mathcal{U},T}^{\textrm{ex}}(C) \leq (1+\varepsilon)p_{\mathcal{U},T}^{\textrm{ex}}(H)$.
% $p_{\mathcal{U},T}(H) = (1 \pm \varepsilon ) p_{\mathcal{U},T}(F) > 0$ for each 
%\underbar{bipartite} $k$-matching $F$ with $H \subseteq F \subseteq C \times D$. 
%If $p_{\mathcal{U},T}^{\textrm{ex}}(H) \leq 2^{-\delta n}$, then $(T,H)$ is $U$-small. 
%If $(T,H)$ is neither Otherwise $(T,H)$ is called \emph{$U$-bad}.
\end{definition}
 %We want to mention that both, the quantity $p_{\mathcal{U},T}(H)$ and the notion of $\mathcal{U}$-goodness only depend on the nodes $V(H) \cap C$ and not 
%the edges $H$ itself. If, by a slight abuse of notation, we abbreviate $p_{\mathcal{U},T}(V(H) \cap C) := p_{\mathcal{U},T}(H)$, then 
%the pair $(T,H)$ with core nodes $c := V(H) \cap C$ is $\mathcal{U}$-good
%if and only if $p_{\mathcal{U},T}(c) = (1 \pm \varepsilon) \cdot p_{\mathcal{U},T}(C)$. 
Again, an alternative characterization for $(T,H)$ with $c = V(H) \cap C$ being $\mathcal{U}$-good 
is the following: if we draw a random cut $U \sim \{ U' \in \mathcal{U}(T) \mid c \subseteq U'\}$, 
then a $(\frac{1}{2} \pm \Theta(\varepsilon))$-fraction has $U \cap C = c$ and the other
$(\frac{1}{2} \pm \Theta(\varepsilon))$-fraction has $U \cap C = C$ (recall that every cut $U \in \mathcal{U}(T)$ has
$|U \cap C| \in \{ 3,k\}$). Once more, this characterization can be derived using conditional probabilities, 
see Cor.~\ref{cor:Pseudo2}.
%We should remark that due to the choice of $t = |U|$, all cuts $U \in \mathcal{U}(T)$
%intersect $C$ in either 3 or in $k$ nodes anyway.

If $(T,H)$ is both $\mathcal{M}$-good and $\mathcal{U}$-good, then it is called \emph{good}. 
There is another type of pairs that will not cause any problems for our analysis: 
If either $p_{\mathcal{M},T}^{\textrm{ex}}(H) \leq 2^{-\delta m}$ or $p_{\mathcal{U},T}^{\textrm{ex}}(H) \leq 2^{-\delta m}$, then we say that $(T,H)$
is \emph{small}. Intuitively, small pairs will only contribute $2^{-\delta m}$ to $\mu_3(\mathcal{R})$ anyway. %\footnote{In our definition, actually a pair could be both good and small --- it doesn't matter

Unfortunately, those two categories will not cover all cases. Hence, if a pair $(T,H)$ is neither 
good nor small, then we call it \emph{bad}. 
%If $(T,H)$ is neither good, nor small, then it is called \emph{bad}. Again, we call $(T,H)$ $M$-bad 
%if the pair is neither small not $\mathcal{M}$-good. Similarly we define \emph{$U$-bad}.
We will use the 0/1 indicator variables ${\tt GOOD}(T,H)$, ${\tt{SMALL}}(T,H)$ and ${\tt BAD}(T,H)$ for the corresponding events.
This allows to split the measure $\mu_3(\mathcal{R})$ into 3 parts
\begin{eqnarray*}
 \mu_3(\mathcal{R}) &=& \E_T \Big[ \E_{|H| = 3} \Big[ p_{\mathcal{M},T}^{\textrm{ex}}(H) \cdot p_{\mathcal{U},T}^{\textrm{ex}}(H) \Big] \Big] \\
 &\leq& \E_T \Big[ \E_{|H| = 3} \Big[ \left( {\tt GOOD}(T,H) + {\tt SMALL}(T,H) + {\tt BAD}(T,H) \right) \cdot p_{\mathcal{M},T}^{\textrm{ex}}(H) \cdot p_{\mathcal{U},T}^{\textrm{ex}}(H) \Big] \Big] \\
 &\stackrel{(*)}{\leq}& \underbrace{\frac{200}{k^2} \cdot\mu_k(\mathcal{R})}_{\textrm{for good pairs}} + 
\underbrace{ 2^{-\delta m}}_{\textrm{for }\textrm{small pairs}} +
\underbrace{ \varepsilon \cdot \mu_3(\mathcal{R}) }_{\textrm{for }\textrm{bad pairs}}
\end{eqnarray*}
for some constant $\delta := \delta(k,\varepsilon) > 0$. 
Rearranging terms and choosing $\varepsilon \leq \frac{1}{2}$ gives the claim of Lemma~\ref{lem:QuadraticGrowOfMeasure}. 
We will spend the rest of Section~\ref{sec:QuadraticMeasureIncrease} justifying the three inequalities that we used in $(*)$.

The estimate that is easy to see is the one concerning the contribution of small pairs. 
We have
\[
 \E_T \Big[ \E_{|H|=3}\Big[  {\tt SMALL}(T,H) \cdot p_{\mathcal{M},T}^{\textrm{ex}}(H) \cdot p_{\mathcal{U},T}^{\textrm{ex}}(H) \Big] \Big] \leq 2^{-\delta m}
\]
because each time that ${\tt SMALL}(T,H) = 1$, we have by definition $p_{\mathcal{M},T}^{\textrm{ex}}(H) \cdot p_{\mathcal{U},T}^{\textrm{ex}}(H) \leq 2^{-\delta m}$.

%the inequality $(*)$ is clear for the small pairs --- we will prove $(*)$ for good and bad pairs
%in the in the remainder of this paper and $\delta := \delta(k,\varepsilon) > 0$ will be some constant. 
% and $m$ is large enough. 
%of Moving both $\varepsilon \cdot \mu_3(\mathcal{R})$
%Note that rearranging then yields
%\[
%  \mu_k(\mathcal{R}) \geq \frac{k^2}{200} \cdot \Big( (1-2\varepsilon) \cdot \mu_3(\mathcal{R}) - 2^{-\delta m} \Big)
%\geq \frac{k^2}{400} \cdot \mu_3(\mathcal{R}) - 2^{-\frac{\delta}{2} m},
%\]

\subsection{Contribution of good pairs}

We continue with bounding the contribution of good pairs, i.e. we will show that 
\begin{equation} \label{eq:contributionGoodPartitions}
  \E_T \Big[ \E_{|H| = 3} \Big[ {\tt GOOD}(T,H) \cdot p_{\mathcal{M},T}^{\textrm{ex}}(H) \cdot p_{\mathcal{U},T}^{\textrm{ex}}(H) \Big] \Big] \leq \frac{200}{k^2} \cdot \mu_k(\mathcal{R}).
\end{equation}
for $\varepsilon > 0$ small enough. The reason, why one should expect a $\frac{1}{k^2}$ term is based on the  insight that only a $O(\frac{1}{k^2})$
fraction of 3-matchings $H$ can actually give a \emph{positive} contribution to 
the LHS of \eqref{eq:contributionGoodPartitions}. The reason is that if we had good pairs $(T,H)$
and $(T,H^*)$ where $H$ and $H^*$ do not share 2 edges, then we could find a slack-0 entry in $\mathcal{R}$.
In fact, this subsection contains the core arguments, \emph{why} the matching polytope has no compact LP
representation. It is also the part where we make use of the combinatorial properties of
matchings and cuts.

% \begin{lemma}
% Let $T$ be any partition and let $H'$ be a perfect matching in the bipartite graph $C \times D$. Then
% \[
%  p_{\mathcal{M},T}(H') \cdot p_{\mathcal{U},T}(H') \geq \Omega(k^2) \cdot \E_{H \in {H' \choose 3}} \Big[ {\tt GOOD}(T,H) \cdot p_{\mathcal{M},T}(H) \cdot p_{\mathcal{U},T}(H) \Big]
% \]
% \end{lemma}
% \begin{proof}
% First note that for any pair $(T,H)$ with $H \subseteq H'$ that is good, we have
% $p_{\mathcal{M},T}(H') \geq (1-\varepsilon) \cdot p_{\mathcal{M},T}(H)$ and $p_{\mathcal{U},T}(H') \geq (1-\varepsilon) p_{\mathcal{U},T}(H)$, hence
% \[
%  \E_{H \in {H' \choose 3}} \Big[ {\tt GOOD}(T,H) \cdot p_{\mathcal{M},T}(H) \cdot p_{\mathcal{U},T}(H) \Big]
% \leq (1+\varepsilon)^2 \cdot p_{\mathcal{M},T}(H') \cdot p_{\mathcal{M},T}(H') \cdot \E_{H \sim {H' \choose 3}}[ GOOD(T,H)]
% \]
% It remains to show that only a $O(\frac{1}{k^2}$-fraction of pairs $(T,H)$ can be good.
% \[
%  p_{\mathcal{M},T}(H') \cdot p_{\mathcal{U},T}(H') \geq (1-\varepsilon)^2 \cdot \E_{H \sim }[ GOOD( ) \cdot ]
% \]
% Suppose that there is at least one 3-matching $H^* \subseteq H'$ so that $(T,H^*)$ is good, since 
% otherwise the right hand side is 0 anyway. Let $c^* := V(H^*) \cap C$ be the 3 core nodes. 
% \end{proof}

%More formally, let ${\tt POS}(T,H)$ be the indicator variable telling whether the contribution of pair $(T,H)$ to the left hand side of \eqref{eq:contributionGoodPartitions}
%is positive. In other words, ${\tt POS}(T,H) = 1$ if and only if
%${\tt GOOD}(T,H) \cdot p_{\mathcal{M},T}(H) \cdot p_{\mathcal{U},T}(H) > 0$.
\begin{lemma} \label{lem:FractionPositiveContrGoodPartitions}
For any partition $T$ and any $k$-matching $F \subseteq C \times D$, one has $\Pr_{H \sim {F \choose 3}}[ {\tt GOOD}(T,H) =1 ] \leq \frac{100}{k^2}$.
\end{lemma}
\begin{proof}
Consider pairs $(T,H)$ and $(T,H^*)$ with $H,H^* \subseteq F$ that are both good.
We claim that then $|H \cap H^*| \geq 2$. For the sake of contradiction suppose that $|H \cap H^*| \leq 1$.
Then there are distinct nodes $u,v \in V(H \backslash H^*) \cap C$.
Now arbitrarily extend $H^*$ to a $k$-matching $F^*$ with $H^* \cup \{ (u,v) \} \subseteq F^* \subseteq E(C \cup D)$.
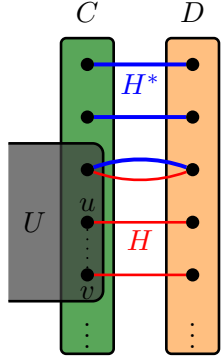
\begin{center}
\psset{unit=0.7cm}
\begin{pspicture}(0,0)(3,6.5)
\drawRect{linearc=0.1,linewidth=1pt,linecolor=black,fillcolor=green!30!gray,fillstyle=solid}{0}{0}{1}{6}
\drawRect{linearc=0.1,linewidth=1pt,linecolor=black,fillcolor=orange!50!white,fillstyle=solid}{2}{0}{1}{6}
\psline[linewidth=1pt,linestyle=solid,linecolor=black,fillstyle=solid,fillcolor=darkgray,opacity=0.7,linearc=0.2](-1,4)(0.8,4)(0.8,1)(-1,1) \rput[c](-0.5,2.5){$U$}
\cnode*(0.5,5.5){2.5pt}{c1}
\cnode*(0.5,4.5){2.5pt}{c2}
\cnode*(0.5,3.5){2.5pt}{c3}
\cnode*(0.5,2.5){2.5pt}{c4}
\cnode*(0.5,1.5){2.5pt}{c5}
\cnode*(2.5,5.5){2.5pt}{d1}
\cnode*(2.5,4.5){2.5pt}{d2}
\cnode*(2.5,3.5){2.5pt}{d3}
\cnode*(2.5,2.5){2.5pt}{d4}
\cnode*(2.5,1.5){2.5pt}{d5}
\ncline[linecolor=blue,linewidth=1.5pt]{c1}{d1} \nbput[labelsep=2pt]{\blue{$H^*$}}
\ncline[linecolor=blue,linewidth=1.5pt]{c2}{d2}
\ncarc[linecolor=blue,linewidth=1.5pt,arcangle=20]{c3}{d3}
\ncarc[linecolor=red,linewidth=1.0pt,arcangle=-20]{c3}{d3}
\ncline[linecolor=red,linewidth=1.0pt]{c4}{d4} \nbput[labelsep=2pt]{\red{$H$}}
\ncline[linecolor=red,linewidth=1.0pt]{c5}{d5}
\rput[c](0.5,0.5){$\vdots$}
\rput[c](2.5,0.5){$\vdots$}
\nput[labelsep=2pt]{90}{c4}{$u$}
\nput[labelsep=2pt]{-90}{c5}{$v$}
\ncline[linestyle=dotted,linecolor=black]{c4}{c5}
\rput[c](0.5,6.5){$C$}
\rput[c](2.5,6.5){$D$}
\end{pspicture}
\end{center}
By assumption  $(T,H^*)$ is good, hence $p_{\mathcal{M},T}(F^*) \geq \frac{1}{1+\varepsilon} p_{\mathcal{M},T}(H^*) > 0$. In other words, 
there exists a matching $M \in \mathcal{M}(T)$ with $(u,v) \in M$.
But also $(T,H)$ is good and hence $p_{\mathcal{U},T}^{\textrm{ex}}(H) > 0$, which implies that there 
is a cut $U \in \mathcal{U}(T)$ so that $U \cap C = V(H) \cap C$. 
Then $(u,v)$ runs inside of $U$ and hence $|\delta(U) \cap M| = 1$, which is a contradiction to $\mu_1(\mathcal{R}) = 0$.
Thus good pairs must indeed overlap in at least 2 edges. 

%With this insight, we can easily bound
Now fix an $H^*$ so that $(T,H^*)$ is good (if there is none, there is nothing to show). 
Then
 \[
\Pr_{H \sim {F \choose 3}}[{\tt GOOD}(T,H) = 1] \leq \Pr_{H \sim {F \choose 3}}[ |H \cap H^*| \geq 2] \leq \frac{3k}{{k \choose 3}} \leq \frac{100}{k^2}.
\] 
This settles the claim.
\end{proof}

%\begin{lemma}
%Suppose that $(T,H)$ is a good pair and $H' \supseteq H$. Then $p_{\mathcal{M},T}(H) \cdot p_{\mathcal{U},T}(H) \leq (1+3\varepsilon) \cdot p_{\mathcal{U},T}(H') \cdot p_{\mathcal{M},T}(H')$.
%\end{lemma}

Now we can easily relate the contribution of the good pairs with the quantity $\mu_k(\mathcal{R})$. In particular 
we use that by definition, for a good pair $(T,H)$ and any $k$-matching $F$ with $H \subseteq F \subseteq C \times D$
one has both  $p_{\mathcal{M},T}^{\textrm{ex}}(H) \leq (1+\varepsilon)p_{\mathcal{M},T}(H) \leq (1+\varepsilon)^2 p_{\mathcal{M},T}^{\textrm{ex}}(F)$ (see Lemma~\ref{lem:pEx-vs-p})  
and $p_{\mathcal{U},T}^{\textrm{ex}}(H) \leq (1+\varepsilon) p_{\mathcal{U},T}^{\textrm{ex}}(C) = (1+\varepsilon) p_{\mathcal{U},T}^{\textrm{ex}}(F)$. The inequality in \eqref{eq:contributionGoodPartitions} then follows from
\begin{eqnarray*}
& & \E_T\Big[  \E_{|H| = 3}[ {\tt GOOD}(T,H) \cdot p_{\mathcal{M},T}^{\textrm{ex}}(H) \cdot p_{\mathcal{U},T}^{\textrm{ex}}(H) ] \Big] \\
&=& \E_T \Big[ \E_{|F| = k} \Big[ \E_{H \sim {F \choose 3}} \Big[ {\tt GOOD}(T,H) \cdot 
  \underbrace{p_{\mathcal{M},T}^{\textrm{ex}}(H)}_{\leq (1+\varepsilon)^2 p_{\mathcal{M},T}^{\textrm{ex}}(F)} \cdot \underbrace{p_{\mathcal{U},T}^{\textrm{ex}}(H)}_{\leq (1+\varepsilon) p_{\mathcal{U},T}^{\textrm{ex}}(F)} \Big] \Big] \Big] \\
&\leq& 2 \cdot \underbrace{\E_T \Big[ \E_{|F| = k} \Big[ p_{\mathcal{M},T}^{\textrm{ex}}(F) \cdot p_{\mathcal{U},T}^{\textrm{ex}}(F)}_{=\mu_k(\mathcal{R})} \cdot \underbrace{\E_{H \sim {F \choose 3}} [{\tt GOOD}(T,H)]}_{\leq 100/k^2} \Big] \Big] \\
&\leq& \frac{200}{k^2} \cdot \mu_k(\mathcal{R}).
\end{eqnarray*}
Here we assume that $\varepsilon \leq \frac{1}{8}$.

\subsection{The pseudo-random behavior of large sets\label{sec:PseudoRandomBehavior}}

It remains to bound the contribution of bad pairs. 
Before we continue with that, %with bounding the contribution of bad partitions, 
we want to describe a general phenomenon concerning the distribution of large set families.
To give an example, suppose you have a set family $X \subseteq 2^{[m]}$ with $|X| \geq 2^{(1-\varepsilon)m}$ for a small
enough constant $\varepsilon>0$. Then  $99\%$ of indices $i \in \{ 1,\ldots,m\}$ will be
in $50\% \pm 1\%$ of sets in the family $X$. 

Well, we need a slightly more general statement which we will prove using 
an \emph{entropy counting argument}. Recall that for a random variable $y$ over $\{1,\ldots,k\}$, the 
\emph{entropy} is defined by $H(y) := \sum_{j=1}^k \Pr[y=j] \cdot \log_2 \frac{1}{\Pr[y=j]}$.
Moreover, the entropy is maximized if $y$ is drawn from the \emph{uniform distribution}; in that case we have
$H(y) = \log_2(k)$. A useful property is that entropy is \emph{sub-additive}. For example
if $y = (y_1,\ldots,y_m)$ is a random vector, then $H(y) \leq \sum_{i=1}^m H(y_i)$. In the following, if we write $y \sim Y$,
then $y$ is a uniformly drawn random element from $Y$.
We need a crucial lemma that appeared already in a less general form in Razborov's paper~\cite{Disjointness-Razborov90}:
\begin{lemma} \label{lem:PseudoRandomBehavior}
For all $\varepsilon > 0$ and $q \in \setN$, there is a constant $\delta := \delta(\varepsilon,q) > 0$ so that the following is true: 
Take finite sets $X_1,\ldots,X_m$ with $1 \leq |X_i| \leq q$ for $i=1,\ldots,m$ and denote $X := X_1 \times \ldots \times X_m$.
Let $Y \subseteq X$ be a subset of size $|Y| \geq 2^{-\delta m}|X|$. An index $i \in [m]$ is called \emph{$\varepsilon$-unbiased}, if 
\[
  \frac{1}{1+\varepsilon} \cdot \frac{1}{|X_i|} \leq \Pr_{y \sim Y}[y_i = j] \leq (1+\varepsilon) \cdot \frac{1}{|X_i|} \quad \forall j \in X_i.
\] 
Then at most $\varepsilon m$ many indices will be $\varepsilon$-biased.
%Then $\Pr_{i \in [m]}[\varepsilon\textrm{-biased}(i)] \leq \varepsilon$.
%F\"ur jede Indexmenge $J \subseteq [m]$ mit $|J| \geq \varepsilon m$ gilt
%$ \E_{i \sim J}[ \textrm{biased}(i)] \leq \varepsilon$
\end{lemma}
\begin{proof}
We consider the random variable $y \sim Y$ and fix an index $i$. %Sei $h(x) := x \cdot \log_2(\frac{1}{x})$. 
%Wenn wir $p_j = \Pr[y_i=j]$ abk\"urzen, 
The entropy of the $i$th coordinate is
\begin{equation}  \label{eq:EntropyOfSingleCoord}
  H(y_i)  = \sum_{j \in X_i} \Pr[y_i=j] \cdot \log_2 \left(\frac{1}{\Pr[y_i=j]}\right) \leq \log_2(|X_i|).
\end{equation}
This bound follows from Jensen's inequality and the observation that the $\log_2$ function is concave.
In fact, the $\log_2$ function is \emph{strictly} concave, which means that the inequality in 
\eqref{eq:EntropyOfSingleCoord} is tight only if 
 $\Pr[y_i=j] = \frac{1}{|X_i|}$ for all outcomes $j \in X_i$. 
In particular, if $i$ is biased in the sense  of the definition above, 
then  $H(y_i) < \log_2(|X_i|)$. The entropy function $H$ is continuous, hence for compactness
reasons there has to be a constant $c := c(\varepsilon,q)>0$ so that $H(y_i) \leq \log_2(|X_i|) - c$ 
holds for each $\varepsilon$-biased index $i$.

Now, we assume for the sake of contradiction that there are $\varepsilon m$ indices that
are $\varepsilon$-biased. Then we can bound the entropy of $y \sim Y$ by
\[
 \log_2(|Y|) = H(y) \stackrel{\textrm{subadd.}}{\leq} \sum_{i=1}^m H(y_i)  \leq \sum_{i=1}^m \log_2(|X_i|) - c\varepsilon m = \log_2(|X|) - c\varepsilon m 
\]
Rearranging yields  $\frac{|Y|}{|X|} \leq 2^{-c \varepsilon m}$, which contradicts the assumption
if we choose $\delta < c\varepsilon$.
\end{proof}

There is an equivalent way of stating ``unbiasedness'' that is closer to our definition
of good pairs: %will be useful for our approach:
%One can reinterpret what the implication of an unbiased index as follows:
\begin{corollary} \label{cor:Pseudo1}
Suppose we have sets $Y \subseteq X = X_1 \times \ldots \times X_m$. For each index $i$ that 
is % as in Lemma~\ref{lem:PseudoRandomBehavior} and an index $i$ that is 
 $\varepsilon$-unbiased and each $j \in X_i$ one has
\[
 \frac{1}{1+\varepsilon} \Pr_{y \sim X}[y \in Y] \leq \Pr_{y \sim X}[y \in Y \mid y_i=j] \leq (1+\varepsilon) \Pr_{y \sim X}[y \in Y].
\]
\end{corollary}
\begin{proof}
We simply rewrite the conditional probability as
\[
  \Pr_{y \sim X}[y \in Y \mid y_i = j] = \underbrace{\Pr_{y \sim Y}[y_i=j]}_{\leq (1+\varepsilon)/|X_i|} \cdot \frac{\Pr_{y \sim X}[y \in Y]}{\underbrace{\Pr_{y \sim X}[y_i=j]}_{=1/|X_i|}}
 \leq (1+\varepsilon) \cdot \Pr_{y \sim X}[y \in Y].
\]
using Bayes' Theorem\footnote{Recall that \emph{Bayes' Theorem} says that
$\Pr[A \mid B] = \frac{\Pr[A \cap B]}{\Pr[B]} = \frac{\Pr[A \cap B]}{\Pr[A]} \cdot \frac{\Pr[A]}{\Pr[B]} = \Pr[B \mid A] \cdot \frac{\Pr[A]}{\Pr[B]}$}. The other direction is analogous. 
\end{proof}
We will apply Lemma~\ref{lem:PseudoRandomBehavior} twice in our proof: once for cuts and once for matchings. 
The cuts all have the same size, so the probability distribution for cuts is not a product distribution.
Hence, we need a slight modification:
\begin{corollary} \label{cor:Pseudo2}
Let $Y \subseteq Z \subseteq X = \{ 0,1\}^{m}$ be sets with  $m$ even and $Z := \{ y \in \{ 0,1\}^m \mid \|y\|_1 = \frac{m}{2}\}$.  
%If index $i$ is $\varepsilon$-unbiased in the sense of Lemma~\ref{lem:PseudoRandomBehavior}, then
Then for each index $i$ that is $\varepsilon$-unbiased with respect to $y \sim Y$ and each $j \in X_i$ one has
\[
 \frac{1}{1+\varepsilon} \Pr_{y \sim Z}[y \in Y] \leq \Pr_{y \sim Z}[y \in Y \mid y_i=j] \leq (1+\varepsilon) \Pr_{y \sim Z}[y \in Y]
\]
\end{corollary}
\begin{proof}
We can use the same estimate as in Cor.~\ref{cor:Pseudo1}, and observe that still $\Pr_{y \sim Z}[y_i=j] = \frac{1}{2}$ for $j \in \{ 0,1\}$
even though the uniform distribution over $Z$ is not a product distribution set.
\end{proof}

\subsection{Contribution of bad pairs}

We finally continue with proving that the contribution of the bad pairs is bounded\footnote{We should mention that the analysis for the bad pairs differs from the STOC'14 version.}.
Recall that our goal will be to show that
\begin{equation} \label{eq:ContrBadPart1}
  \E_T \Big[ \E_{|H|=3} \Big[ {\tt BAD}(T,H) \cdot p_{\mathcal{M},T}^{\textrm{ex}}(H) \cdot p_{\mathcal{U},T}^{\textrm{ex}}(H) \Big] \Big] \stackrel{!}{\leq} \varepsilon \cdot \mu_3(\mathcal{R}).
\end{equation}
First, observe that the left hand side of \eqref{eq:ContrBadPart1} %can be written equivalently as
can be stated in the following more explicit form
\begin{equation} \label{eq:ContrBadPart2}
% \E_T \Big[ \E_{|H|} \Big[ {\tt BAD}(T,H) \cdot p_{\mathcal{M},T}^{\textrm{ex}}(H) \cdot p_{\mathcal{U},T}^{\textrm{ex}}(H) \Big] \Big] 
 \E_{\substack{T=(A,C,D,B) \\ H \subseteq C \times D: |H| = 3}} \Big[ \E_{\substack{U \sim \mathcal{U}_{\textrm{all}}(T): U \cap C = V(H) \cap C \\ M \sim \mathcal{M}_{\textrm{all}}(T): M \cap (C \times D) = H}} \Big[ {\tt BAD}(T,H) \cdot \bm{1}_{(U,M) \in \mathcal{R}}\Big] \Big]
\end{equation}
Here $1_{\mathcal{E}}$ denotes the indicator variable for an event $\mathcal{E}$.
In \eqref{eq:ContrBadPart2} we pick first the pair $(T,H)$ and then $(U,M)$. 
Now we want to switch the expectations. 
For notational convenience, define 
\[
  \mathcal{P}(U,M) := \left\{ T=(A,C,D,B) \mid \begin{array}{c} U \in \mathcal{U}_{\textrm{all}}(T) \textrm{ with }  U \cap C = V(H) \cap C \textrm{ and} \\   M \in \mathcal{M}_{\textrm{all}}(T) \textrm{ with }  M \cap (C \times D) = H \end{array} \right\}
\]
as all the partitions that are compatible with the pair $(U,M) \in Q_3$, where $H := \delta(U) \cap M$. Then switching the expectation gives
\[
 \eqref{eq:ContrBadPart2} = \E_{\substack{(U,M) \in Q_3 \\ H := \delta(U) \cap M}} \Big[ \bm{1}_{(U,M) \in \mathcal{R}} \cdot \underbrace{\E_{T \sim \mathcal{P}(U,M)}[ {\tt BAD}(T,H)]}_{\textrm{claim: }\leq \varepsilon} \Big] \stackrel{!}{\leq} \varepsilon \cdot \E_{(U,M) \in Q_3}[ \bm{1}_{(U,M) \in \mathcal{R}}] = \varepsilon \cdot \mu_3(R).
\]
Implicitly we have been using here that for symmetry reasons the cardinality $\mathcal{P}(U,M)$ is the same for all $(U,M) \in Q_3$.
As indicated in the above formula, it suffices to prove that for each pair $(U,M) \in Q_3$, 
only an $\varepsilon$-fraction of compatible partitions $T$ can have the pair $(T,H)$
being bad where $H$ is uniquely defined by $(U,M)$. 
This then implies \eqref{eq:ContrBadPart1}.
Formally, we will prove:
\begin{lemma}  \label{lem:UMbadWithProbEpsilon}
Fix any $(U,M) \in Q_3$ and abbreviate $H := \delta(U) \cap M$. 
Then
\[
  \Pr_{T \sim \mathcal{P}(U,M)}[{\tt BAD}(T,H) = 1] \leq \varepsilon.
\]
\end{lemma}
%Intuitively, this Lemma says that just
Recall that we did call a pair $(T,H)$ bad if it is neither small nor good. Now, let us 
distinguish whether $(T,H)$ is bad because of the cut part or because of the matching part. 
We define $(T,H)$ to be \emph{$\mathcal{U}$-bad} if it is neither small nor $\mathcal{U}$-good. Similarly
we call $(T,H)$ $\mathcal{M}$-bad if it is neither small nor $\mathcal{M}$-good. 
In the next subsections we will separately show that for all $(U,M) \in Q_3$ and $H = \delta(U) \cap M$ one has
\[
 \Pr_{T \sim \mathcal{P}(U,M)}[\mathcal{U}{\tt{-BAD}}(T,H)=1] \leq \frac{\varepsilon}{2} \quad \textrm{and} \quad \Pr_{T \sim \mathcal{P}(U,M)}[\mathcal{M}{\tt-BAD}(T,H)=1] \leq \frac{\varepsilon}{2}  
\]
which then concludes Lemma~\ref{lem:UMbadWithProbEpsilon} using the union bound.

\subsubsection{The contribution for $\mathcal{U}$-bad pairs}

First, we show that for any pair $(U,M)$, most compatible partitions $T$ 
have either $(T,H)$ being small or 
$\mathcal{U}$-good. If we have a perfect matching $M$ and a nodeset $W$, then we say that $M$ 
\emph{crosses} $W$ if $M \cap \delta(W) \neq \emptyset$. 
%Moreover, we say that $W$ does not cross a cut $U$ if either $W \subseteq U$ or $W \cap U = \emptyset.$
\begin{lemma} \label{lem:BoundOnUBadProb}
Fix any pair $(U^*,M^*) \in Q_3$ and abbreviate $H := \delta(U^*) \cap M^*$. 
Then
\[ \Pr_{T \sim \mathcal{P}(U^*,M^*)}[\mathcal{U}{\tt{-BAD}}(T,H)=1] \leq \frac{\varepsilon}{2}. \]
\end{lemma}
\begin{proof}
We imagine that we generate the random partition
 $T \sim \mathcal{P}(U^*,M^*)$ in two phases. 
In the first phase, we randomly partition the nodes $V \setminus V(H)$ into disjoint blocks
 $B_1,\ldots,B_m $ each of size $2(k-3)$; $k-3$ nodes $D \setminus V(H)$ and blocks $\tilde{A}_1,\ldots,\tilde{A}_{m+1}$, each of size
$k-3$. Here we choose the blocks conditioning on the following events: (1) $M^*$ does not cross any of those blocks; 
(2) $U^* \subseteq V \setminus (B \cup D)$ and (3) $|U^* \cap \tilde{A}_i| \in \{ 0, |\tilde{A}_i|\}$ for all $i=1,\ldots,m+1$. 
Figure~\ref{fig:BoundUBad} gives a visualization.

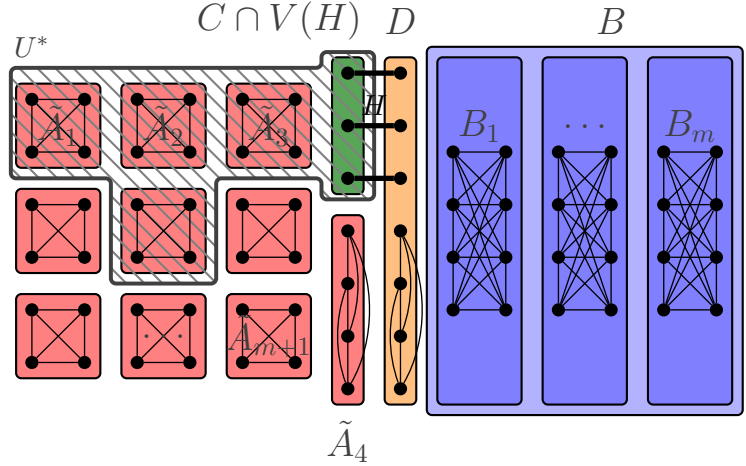
\begin{figure}
\begin{center}
\psset{unit=0.7cm}
\begin{pspicture}(0,-0.5)(14,7.5)
%\drawRect{fillcolor=red!30!white,fillstyle=solid,linearc=0.1}{0}{0}{6}{7} % F
\drawRect{fillcolor=green!30!gray,fillstyle=solid,linearc=0.1}{6.2}{4.2}{0.6}{2.6} % T_0=C
\drawRect{fillcolor=blue!30!white,fillstyle=solid,linearc=0.1}{8.0}{0.0}{6}{7} % whole B
\drawRect{fillcolor=orange!50!white,fillstyle=solid,linearc=0.1}{7.2}{0.2}{0.6}{6.6} % D=B_i^*
%\drawRect{fillcolor=blue!50!white,fillstyle=solid,linearc=0.1}{6.2}{0.2}{1.6}{3.6} %
\drawRect{fillcolor=blue!50!white,fillstyle=solid,linearc=0.1}{8.2}{0.2}{1.6}{6.6} %
\drawRect{fillcolor=blue!50!white,fillstyle=solid,linearc=0.1}{10.2}{0.2}{1.6}{6.6} %
\drawRect{fillcolor=blue!50!white,fillstyle=solid,linearc=0.1}{12.2}{0.2}{1.6}{6.6} %
%\rput[c](2.5,7.5){\Large{\darkgray{$F$}}}
\rput[c](5.2,7.5){\Large{\darkgray{$C \cap V(H)$}}}
\rput[c](7.5,7.5){\Large{\darkgray{$D$}}}
\rput[c](11.5,7.5){\Large{\darkgray{$B$}}}
%\rput[c](7,-0.4){\Large{\darkgray{$B_1$}}}
\rput[c](9,5.5){\Large{\darkgray{$B_1$}}}
\rput[c](11,5.5){\Large{\darkgray{$\ldots$}}}
\rput[c](13,5.5){\Large{\darkgray{$B_{m}$}}}
%\rput[c](6.5,7.3){\Large{\darkgray{$U$}}}
% nodes and edges in F
\drawRect{fillstyle=solid,fillcolor=red!50!white,linearc=0.1}{0.2}{0.7}{1.6}{1.6}
\drawRect{fillstyle=solid,fillcolor=red!50!white,linearc=0.1}{0.2}{2.7}{1.6}{1.6}
\drawRect{fillstyle=solid,fillcolor=red!50!white,linearc=0.1}{0.2}{4.7}{1.6}{1.6}
\drawRect{fillstyle=solid,fillcolor=red!50!white,linearc=0.1}{2.2}{0.7}{1.6}{1.6}
\drawRect{fillstyle=solid,fillcolor=red!50!white,linearc=0.1}{2.2}{2.7}{1.6}{1.6}
\drawRect{fillstyle=solid,fillcolor=red!50!white,linearc=0.1}{2.2}{4.7}{1.6}{1.6}
\drawRect{fillstyle=solid,fillcolor=red!50!white,linearc=0.1}{4.2}{0.7}{1.6}{1.6}
\drawRect{fillstyle=solid,fillcolor=red!50!white,linearc=0.1}{4.2}{2.7}{1.6}{1.6}
\drawRect{fillstyle=solid,fillcolor=red!50!white,linearc=0.1}{4.2}{4.7}{1.6}{1.6}
\drawRect{fillstyle=solid,fillcolor=red!50!white,linearc=0.1}{6.2}{0.2}{0.6}{3.6} 
\pspolygon[fillstyle=vlines,hatchcolor=gray,linecolor=darkgray,linewidth=1.5pt,linearc=0.1](0.1,6.6)(6,6.6)(6,6.9)(7,6.9)(7,4.1)(6,4.1)(6,4.5)(4,4.5)(4,2.5)(2,2.5)(2,4.5)(0.1,4.5) \rput[c](0.5,7.0){$\Large{\darkgray{U^*}}$} % U
%\pspolygon[fillstyle=vlines,hatchcolor=gray,linecolor=darkgray,linewidth=1.5pt,linearc=0.1](0.1,6.6)(6,6.6)(6,6.9)(7,6.9)(7,4.1)(6,4.1)(6,4.5)(4,4.5)(4,2.5)(2,2.5)(2,4.5)(0.1,4.5) % U
\multido{\N=0.5+2.0}{3}{ % F
  \multido{\n=1.0+2.0}{3}{%
     \rput[c](\N,\n){%
        \cnode*(0,0){2.5pt}{A}%
        \cnode*(1,0){2.5pt}{B}%
        \cnode*(0,1){2.5pt}{C}%
        \cnode*(1,1){2.5pt}{D}%
        \ncline[linewidth=0.5pt,linecolor=black]{A}{B}%
       \ncline[linewidth=0.5pt,linecolor=black]{A}{C}%
       \ncline[linewidth=0.5pt,linecolor=black]{A}{D}%
       \ncline[linewidth=0.5pt,linecolor=black]{B}{C}%
       \ncline[linewidth=0.5pt,linecolor=black]{B}{D}%
       \ncline[linewidth=0.5pt,linecolor=black]{C}{D}%
     }
  }
}
\rput[c](1,5.5){\Large{\darkgray{$\tilde{A}_1$}}}
\rput[c](3,5.5){\Large{\darkgray{$\tilde{A}_2$}}}
\rput[c](5,5.5){\Large{\darkgray{$\tilde{A}_3$}}}
\rput[c](6.5,-0.5){\Large{\darkgray{$\tilde{A}_4$}}}
\rput[c](3,1.5){\Large{\darkgray{$\ldots$}}}
\rput[c](5,1.5){\Large{\darkgray{$\tilde{A}_{m+1}$}}}
\psset{linewidth=0.5pt,linecolor=black} % B
\multido{\N=8.5+2.0}{3}{
 \rput[c](\N,2){%
  \cnode*(0,3){2.5pt}{A1}  \cnode*(1,3){2.5pt}{B1}%
  \cnode*(0,2){2.5pt}{A2}  \cnode*(1,2){2.5pt}{B2}%
  \cnode*(0,1){2.5pt}{A3}  \cnode*(1,1){2.5pt}{B3}%
  \cnode*(0,0){2.5pt}{A4}  \cnode*(1,0){2.5pt}{B4}%
  \ncline{A1}{B1}%
  \ncline{A1}{B2}%
  \ncline{A1}{B3}%
  \ncline{A1}{B4}%
  \ncline{A2}{B1}%
  \ncline{A2}{B2}%
  \ncline{A2}{B3}%
  \ncline{A2}{B4}%
  \ncline{A3}{B1}%
  \ncline{A3}{B2}%
  \ncline{A3}{B3}%
  \ncline{A3}{B4}%
  \ncline{A4}{B1}%
  \ncline{A4}{B2}%
  \ncline{A4}{B3}%
  \ncline{A4}{B4}%
  \ncline{A1}{A4}%
  \ncline{B1}{B4}%
 }
}
\multido{\N=6.5+1.0}{2}{% C+D
  \cnode*(\N,3.5){2.5pt}{A}%
  \cnode*(\N,2.5){2.5pt}{B}%
  \cnode*(\N,1.5){2.5pt}{C}%
  \cnode*(\N,0.5){2.5pt}{D}%
  \ncline{A}{D}%
  \ncarc[arcangle=20]{A}{C}%
  \ncarc[arcangle=-20]{B}{D}%
  \ncarc[arcangle=30]{A}{D}%
}%
%\multido{\N=0.5+1.0}{4}{%
%  \multido{\n=0.5+1.0}{4}{%
%     \cnode*(6.5,\N){2.5pt}{A}%
%     \cnode*(7.5,\n){2.5pt}{B}%
%     \ncline{A}{B}%
%  }
%}
%\psline(6.5,0.5)(6.5,3.5)
%\psline(7.5,0.5)(7.5,3.5)
\cnode*(6.5,6.5){2.5pt}{c1} \cnode*(7.5,6.5){2.5pt}{d1} \ncline[linewidth=2pt,linecolor=black]{c1}{d1}
\cnode*(6.5,5.5){2.5pt}{c2} \cnode*(7.5,5.5){2.5pt}{d2} \ncline[linewidth=2pt,linecolor=black]{c2}{d2} \naput[labelsep=2pt]{$H$}
\cnode*(6.5,4.5){2.5pt}{c3} \cnode*(7.5,4.5){2.5pt}{d3} \ncline[linewidth=2pt,linecolor=black]{c3}{d3}
\end{pspicture}
\caption{Situation after we decided for $\tilde{A}_1,\ldots,\tilde{A}_{m+1},H,B_1,\ldots,B_{m}$. The matching $M^*$ is a subset of the depicted edges. \label{fig:BoundUBad}}
\end{center}
\end{figure}

Now, consider the indices $J := \{ i \in [m+1] \mid |U^* \cap \tilde{A}_i| = 0\}$ of blocks
not contained in $U^*$. In particular $|J| = \frac{m+1}{2}$. In the second phase, 
we take a random uniform index $i \sim J$  and declare $\tilde{A}_i$ as the missing $k-3$ nodes  $C \setminus V(H)$. 
Finally, we set $\{A_1,\ldots,A_m\} := \{\tilde{A}_1,\ldots,\tilde{A}_{i-1},\tilde{A}_{i+1},\ldots,\tilde{A}_{m+1}\}$ which completes the description of $T$. 
The two stage process can be summarized to
\[
 \Pr_{T \sim \mathcal{P}(U^*,M^*)}[\mathcal{U}{\tt-BAD}(T,H)=1] = \E_{B_1,\ldots,B_m,D,\tilde{A}_1,\ldots,\tilde{A}_{m+1}} \Big[ \E_{i \sim J}[ \mathcal{U}{\tt-BAD}(T,H)] \Big].
\]
Now fix any outcome of $B_1,\ldots,B_m,D,\tilde{A}_1,\ldots,\tilde{A}_{m+1}$ in the first phase; we will  
show that in any case  $\Pr_{i \sim J}[\mathcal{U}{\tt -BAD}(T,H)=1] \leq \frac{\varepsilon}{2}$.
For a vector $y \in \{ 0,1\}^{m+1}$, we define a cut % index set $I \subseteq [m+1]$, set
\[
  f(y) := (C \cap V(H)) \cup \bigcup_{i: y_i=1} \tilde{A}_i 
\]
and denote $Y := \{ y \in \{ 0,1\}^{m+1} : f(y) \in \mathcal{U} \}$ as well as $Z := \{ y \in \{ 0,1\}^{m+1} : \|y\|_1 = \frac{m+1}{2}\}$ and $X := \{ 0,1\}^{m+1}$.
In other words, the vectors in $Y$ represent all the cuts in the rectangle $\mathcal{R}$
that determine the outcome of $\mathcal{U}{\tt-BAD}(T,H)$ in the second phase.
More precisely, we know that a partition $(T,H)$ will be $\mathcal{U}$-good, if in the second phase we pick an index $i$
so that roughly half of the vectors $y \sim Y$ have $y_i=1$.
%are exaclty the cuts in the rectangle
%$\mathcal{R}$ that 
%and define the family of all candidate cuts
%$
%  \tilde{\mathcal{U}}_{\textrm{all}} = \{ f(I) \mid I \subseteq [m] : |I| = \frac{m+1}{2} \}
%$
%Die Cuts im $\tilde{\mathcal{U}}_{\textrm{all}} := \{ U \in \mathcal{U}_{all} \mid |U \cap \tilde{A}_i| \in \{ 0,|\tilde{A}_i|\}, U \cap (D \cup B) = \emptyset, U \cap C = V(H)\}$ 
%as well as the subset  $\tilde{\mathcal{U}} := \tilde{\mathcal{U}}_{\textrm{all}} \cap \mathcal{U}$ of those cuts that% lie in our rectangle $\mathcal{R}$.
%Recall that in this setting, the pair $(T,H)$ will be $\mathcal{U}$-good, if roughly half of the sets 
%in $\tilde{\mathcal{U}}$ contain the block $\tilde{A}_i$ that will be picked in the second phase. 
Our goal is to use the insight from Section~\ref{sec:PseudoRandomBehavior} to argue that this is the case for most indices. 
We can assume that there is at least one index $i^* \in J$ so that the outcome $(T,H)$
is not small, since otherwise there is nothing to show. 
The partition for that index satisfies $p_{\mathcal{U},T}^{\textrm{ex}}(H) \geq 2^{-\delta m}$ and hence  $|Y| \geq 2^{-\delta m} {m \choose (m+1)/2} \geq 2^{-2\delta m} |X|$ 
for $m$ large enough. 
Now we apply Lemma~\ref{lem:PseudoRandomBehavior} for a small enough $\delta>0$ and obtain that 
at most $\frac{\varepsilon}{4} m$ many indices are $\varepsilon$-biased. 
Then the chance of picking a bad index is at most $\frac{\varepsilon m / 4}{|J|} \leq \frac{\varepsilon}{2}$. Finally, we apply Cor.~\ref{cor:Pseudo2}
to conclude that for
a $(1-\frac{\varepsilon}{2})$-fraction of indices, the corresponding partitions are $\mathcal{U}$-good.
\end{proof}
Admittedly, only now one can fully understand the meaning behind the 
definition of a partition: if we take a partition $(T',H)$, then
the set $C \setminus V(H)$ contains the same number of nodes as any $A_i$.
Assuming that there are enough cuts in $\mathcal{U}' := \{ U \in \mathcal{U}(T') \mid (V(H) \cap C) \subseteq U \}$, 
then we know that for 
most indices $i$, roughly half the cuts of $\mathcal{U}'$ 
will contain $A_i$ --- the other half will not. 
Hence one could imagine to randomly pick an index $i$ and exchange the 
set $C \setminus V(H)$ with  $A_i$ and the emerging pair $(T,H)$ would then be 
$\mathcal{U}$-good with probability $1-\frac{\varepsilon}{2}$. This is a slightly different view on the argument in the proof of Lemma~\ref{lem:BoundOnUBadProb}.

\subsubsection{The contribution for $\mathcal{M}$-bad pairs}

Bounding the chance that a pair $(T,H)$ turns out to be $\mathcal{M}$-bad is fairly similar
to the previous case. Now it will be crucial that $(C \cup D) \setminus V(H)$
has the same size as each block $B_i$ so that we can again play the exchange trick.  
\begin{lemma}
Fix $(U^*,M^*) \in Q_3$ and let $H := \delta(U^*) \cap M^*$. 
Then 
\[
\Pr_{T \sim \mathcal{P}(U^*,M^*)}[\mathcal{M}{\tt{-BAD}}(T,H)=1] \leq \frac{\varepsilon}{2}.
\]
\end{lemma}
\begin{proof}
Again, we imagine that we draw the partition $T$ in two phase. 
In the first phase, we pick blocks  $\tilde{B}_1,\ldots,\tilde{B}_{m+1}$ of each $2(k-3)$ nodes disjoint to $U^*$. 
Then we select blocks $A_1,\ldots,A_m$ of $k-3$ nodes each. 
Again we condition that $M^*$ is not crossing any of those blocks $\tilde{B}_i$ or $A_i$. 
%Knoten die $U$ nicht kreuzen und $U \subseteq (V(H) \cap C) \cup A_1 \cup \ldots \cup A_m$.
Then we randomly partition each $\tilde{B}_i$-block into  $\tilde{B}_i = C_i \dot{\cup} D_i$ with $|C_i| = |D_i| = k-3$
conditioning that $M^* \cap (C_i \times D_i) = \emptyset$ for all $i=1,\ldots,m+1$. 
Figure~\ref{fig:BoundMbad} shows that situation. Note that at this point, we have only determined 3 nodes in each $C$ and $D$.
In the second phase, we select a uniform random index $i \in [m+1]$ and set $C := (C \cap V(H)) \cup C_i$
and $D := (V(H) \cap D) \cup D_i$. Eventually, we set $\{B_1,\ldots,B_m\} := \{ \tilde{B}_1,\ldots,\tilde{B}_{i-1},\tilde{B}_{i+1},\ldots,\tilde{B}_{m+1}\}$ 
which again completes the description of the partition $T$. Observe that for symmetry reasons, the
generated partition $T$ is a uniform element of $\mathcal{P}(U^*,M^*)$.

\begin{figure}
\begin{center}
\psset{unit=0.7cm}
\begin{pspicture}(0,-0.5)(14,7.5)
\drawRect{fillcolor=red!30!white,fillstyle=solid,linearc=0.1}{0}{0}{6}{7} % F
\drawRect{fillcolor=green!30!gray,fillstyle=solid,linearc=0.1}{6.2}{4.2}{0.6}{2.6} % T_0=C
%\drawRect{fillcolor=blue!30!white,fillstyle=solid,linearc=0.1}{8.0}{0.0}{6}{7} % whole B
\drawRect{fillcolor=orange!50!white,fillstyle=solid,linearc=0.1}{7.2}{4.2}{0.6}{2.6} % D=B_i^*
% \tilde{B}_i's
\drawRect{fillcolor=blue!65!white,fillstyle=solid,linearc=0.1}{6.2}{0.2}{1.6}{3.6} %
\drawRect{fillcolor=blue!60!white,fillstyle=solid,linearc=0.1}{8.2}{0.2}{1.6}{6.6} %
\drawRect{fillcolor=blue!60!white,fillstyle=solid,linearc=0.1}{10.2}{0.2}{1.6}{6.6} %
\drawRect{fillcolor=blue!60!white,fillstyle=solid,linearc=0.1}{12.2}{0.2}{1.6}{6.6} %
% C(i)'s
\drawRect{fillcolor=blue!35!white,fillstyle=solid,linearc=0.1,linestyle=none}{7.0}{0.2}{0.8}{3.6} %
\pspolygon[fillcolor=blue!35!white,fillstyle=solid,linearc=0.1,linestyle=none](8.2,0.2)(8.2,4.5)(9,4.5)(9,2.5)(9.8,2.5)(9.8,0.2)
\pspolygon[fillcolor=blue!35!white,fillstyle=solid,linearc=0.1,linestyle=none](10.2,2.5)(11,2.5)(11,4.5)(11.8,4.5)(11.8,6.8)(10.2,6.8)
\pspolygon[fillcolor=blue!35!white,fillstyle=solid,linearc=0.1,linestyle=none](12.2,0.2)(12.2,3.5)(13.8,3.5)(13.8,0.2)
% draw again borders of \tilde{B}_i's
\drawRect{fillcolor=blue!65!white,fillstyle=none,linearc=0.1}{6.2}{0.2}{1.6}{3.6} %
\drawRect{fillcolor=blue!60!white,fillstyle=none,linearc=0.1}{8.2}{0.2}{1.6}{6.6} %
\drawRect{fillcolor=blue!60!white,fillstyle=none,linearc=0.1}{10.2}{0.2}{1.6}{6.6} %
\drawRect{fillcolor=blue!60!white,fillstyle=none,linearc=0.1}{12.2}{0.2}{1.6}{6.6} %
\rput[c](2.5,7.5){\Large{\darkgray{$A$}}}
\rput[c](5.2,8.5){\rnode{Clabel}{\Large{\darkgray{$C \cap V(H)$}}}}
\rput[c](8.8,8.5){\rnode{Dlabel}{\Large{\darkgray{$D \cap V(H)$}}}}
%\rput[c](11.5,7.5){\Large{\darkgray{$B$}}}
\rput[c](7,-0.4){\Large{\darkgray{$\tilde{B}_1$}}}
\rput[c](9,7.5){\Large{\darkgray{$\tilde{B}_2$}}}
\rput[c](11,7.5){\Large{\darkgray{$\ldots$}}}
\rput[c](13,7.5){\Large{\darkgray{$\tilde{B}_{m+1}$}}}
\rput[c](9,6.0){\darkgray{$C_2$}} \rput[c](9,1.0){\darkgray{$D_2$}}
\rput[c](11,6.0){\darkgray{$\ldots$}} \rput[c](11,1.0){\darkgray{$\ldots$}}
\rput[c](13,6.0){\darkgray{$C_{m+1}$}} \rput[c](13,1.0){\darkgray{$D_{m+1}$}}
%\rput[c](6.5,7.3){\Large{\darkgray{$U$}}}
% nodes and edges in F
\drawRect{fillstyle=solid,fillcolor=red!50!white,linearc=0.1}{0.2}{0.7}{1.6}{1.6}
\drawRect{fillstyle=solid,fillcolor=red!50!white,linearc=0.1}{0.2}{2.7}{1.6}{1.6}
\drawRect{fillstyle=solid,fillcolor=red!50!white,linearc=0.1}{0.2}{4.7}{1.6}{1.6}
\drawRect{fillstyle=solid,fillcolor=red!50!white,linearc=0.1}{2.2}{0.7}{1.6}{1.6}
\drawRect{fillstyle=solid,fillcolor=red!50!white,linearc=0.1}{2.2}{2.7}{1.6}{1.6}
\drawRect{fillstyle=solid,fillcolor=red!50!white,linearc=0.1}{2.2}{4.7}{1.6}{1.6}
\drawRect{fillstyle=solid,fillcolor=red!50!white,linearc=0.1}{4.2}{0.7}{1.6}{1.6}
\drawRect{fillstyle=solid,fillcolor=red!50!white,linearc=0.1}{4.2}{2.7}{1.6}{1.6}
\drawRect{fillstyle=solid,fillcolor=red!50!white,linearc=0.1}{4.2}{4.7}{1.6}{1.6}
\pspolygon[fillstyle=vlines,hatchcolor=gray,linecolor=darkgray,linewidth=1.5pt,linearc=0.1](0.1,6.6)(6,6.6)(6,6.9)(7,6.9)(7,4.1)(6,4.1)(6,4.5)(4,4.5)(4,2.5)(2,2.5)(2,4.5)(0.1,4.5) \rput[c](0.5,7.0){\Large{\darkgray{$U^*$}}} % U
%\pspolygon[fillstyle=vlines,hatchcolor=gray,linecolor=darkgray,linewidth=1.5pt,linearc=0.1](0.1,6.6)(6,6.6)(6,6.9)(7,6.9)(7,4.1)(6,4.1)(6,4.5)(4,4.5)(4,2.5)(2,2.5)(2,4.5)(0.1,4.5) % U
\multido{\N=0.5+2.0}{3}{ % F
  \multido{\n=1.0+2.0}{3}{%
     \rput[c](\N,\n){%
        \cnode*(0,0){2.5pt}{A}%
        \cnode*(1,0){2.5pt}{B}%
        \cnode*(0,1){2.5pt}{C}%
        \cnode*(1,1){2.5pt}{D}%
        \ncline[linewidth=0.5pt,linecolor=black]{A}{B}%
       \ncline[linewidth=0.5pt,linecolor=black]{A}{C}%
       \ncline[linewidth=0.5pt,linecolor=black]{A}{D}%
       \ncline[linewidth=0.5pt,linecolor=black]{B}{C}%
       \ncline[linewidth=0.5pt,linecolor=black]{B}{D}%
       \ncline[linewidth=0.5pt,linecolor=black]{C}{D}%
     }
  }
}
\rput[c](1,5.5){\Large{\darkgray{$A_1$}}}
\rput[c](3,5.5){\Large{\darkgray{$A_2$}}}
\rput[c](5,5.5){\Large{\darkgray{$A_3$}}}
\rput[c](3,1.5){\Large{\darkgray{$\ldots$}}}
\rput[c](5,1.5){\Large{\darkgray{$A_m$}}}
\psset{linewidth=0.5pt,linecolor=black} % B
\multido{\N=8.5+2.0}{3}{
 \rput[c](\N,2){%
  \cnode*(0,3){2.5pt}{A1}  \cnode*(1,3){2.5pt}{B1}%
  \cnode*(0,2){2.5pt}{A2}  \cnode*(1,2){2.5pt}{B2}%
  \cnode*(0,1){2.5pt}{A3}  \cnode*(1,1){2.5pt}{B3}%
  \cnode*(0,0){2.5pt}{A4}  \cnode*(1,0){2.5pt}{B4}%
  \ncline{A1}{B1}%
  \ncline{A1}{B2}%
  \ncline{A1}{B3}%
  \ncline{A1}{B4}%
  \ncline{A2}{B1}%
  \ncline{A2}{B2}%
  \ncline{A2}{B3}%
  \ncline{A2}{B4}%
  \ncline{A3}{B1}%
  \ncline{A3}{B2}%
  \ncline{A3}{B3}%
  \ncline{A3}{B4}%
  \ncline{A4}{B1}%
  \ncline{A4}{B2}%
  \ncline{A4}{B3}%
  \ncline{A4}{B4}%
  \ncline{A1}{A4}%
  \ncline{B1}{B4}%
 }
}
\multido{\N=0.5+1.0}{4}{%
  \multido{\n=0.5+1.0}{4}{%
     \cnode*(6.5,\N){2.5pt}{A}%
     \cnode*(7.5,\n){2.5pt}{B}%
     \ncline{A}{B}%
  }
}
\psline(6.5,0.5)(6.5,3.5)
\psline(7.5,0.5)(7.5,3.5)
\cnode*(6.5,6.5){2.5pt}{c1} \cnode*(7.5,6.5){2.5pt}{d1} \ncline[linewidth=2pt,linecolor=black]{c1}{d1}
\cnode*(6.5,5.5){2.5pt}{c2} \cnode*(7.5,5.5){2.5pt}{d2} \ncline[linewidth=2pt,linecolor=black]{c2}{d2} \naput[labelsep=2pt]{$H$}
\cnode*(6.5,4.5){2.5pt}{c3} \cnode*(7.5,4.5){2.5pt}{d3} \ncline[linewidth=2pt,linecolor=black]{c3}{d3}
\pnode(6.5,6.8){C} \ncline[linecolor=darkgray,linewidth=1.5pt,arrowsize=6pt]{<-}{C}{Clabel}
\pnode(7.5,6.8){D} \ncline[linecolor=darkgray,linewidth=1.5pt,arrowsize=6pt]{<-}{D}{Dlabel}
%\psline[linecolor=gray,linewidth=1.5pt](7,3.8)(7,0.2)
\end{pspicture}
\caption{Situation after we decided for $A_1,\ldots,A_m,H,\tilde{B}_1,\ldots,\tilde{B}_{m+1}$ and $\tilde{B}_i = C_i \dot{\cup} D_i$. Again, $M^*$ is a subset of the depicted edges.  \label{fig:BoundMbad}}
\end{center}
\end{figure}
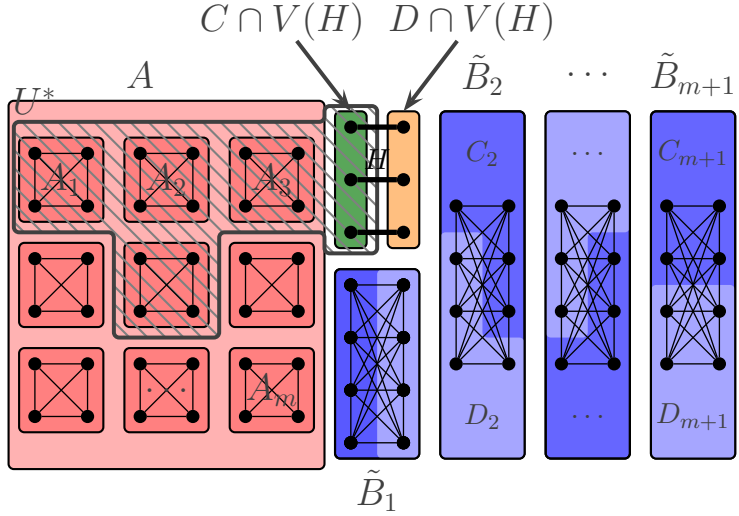

As before, we fix any outcome in the first phase and claim that still 
$\Pr_{i \in [m+1]}[\mathcal{M}{\tt{-BAD}}(T,H)=1] \leq \frac{\varepsilon}{2}$.
Define
\begin{eqnarray*}
  X_i := \{ \textrm{all perfect matchings on }\tilde{B}_i\}  & & \forall i=1,\ldots,m+1  \\
  X_{m+i+1} := \{ \textrm{all perfect matchings on }A_i \} & & \forall i=1,\ldots,m
\end{eqnarray*}
and set $X := X_1 \times \ldots \times X_{2m+1}$. 
We abbreviate $Y := \{ M \in X \mid M \cup H \in \mathcal{M}\}$ as the corresponding
matchings that are in our rectangle $\mathcal{R}$.
Again, we can assume that for some index $i$, the pair $(T,H)$ is not small and hence
$|Y| \geq 2^{-\delta m} \prod_{j \neq i}|X_j| \geq 2^{-2\delta m}|X|$ for $m$ large enough. 

The definition of $\mathcal{M}$-goodness requires that if we draw $M \sim Y$, then the induced
matching that we see on $\tilde{B}_i$ is very close to a uniform random matching. 
Again for a small enough $\delta$, Lemma~\ref{lem:PseudoRandomBehavior}
bounds the number of indices that are $\varepsilon$-biased by $\frac{\varepsilon}{2} m$.
Here $q := \max\{ |X_i| : i=1,\ldots,2m+1\}$ is the number of matchings on $2(k-3)$
nodes, which is some (huge) constant depending on $k$. 
Following Cor.~\ref{cor:Pseudo1}, an unbiased index implies that the corresponding pair
$(T,H)$ will be $\mathcal{M}$-good. This shows the claim.  
\end{proof}
This concludes the proof the Lemma~\ref{lem:QuadraticGrowOfMeasure}
and hence implies our main result, Theorem~\ref{thm:MainTheorem}.

\section{Inapproximability of the matching polytope\label{sec:Inapproximability}}

%Still, let $G = (V,E)$ be the complete graph on $n$ nodes.
%While in this paper, we talked mostly about the convex hull of all perfect matching, 
%it make 
% Let
In this section, we want to discuss the inapproximability for $P_M$, which is 
the convex hull of all matchings in $G$ (not just the perfect ones). 
The polytope $P_M$ has the nice property that it is \emph{monotone}, which means
that for $x \in P_M$ and $\bm{0} \leq y \leq x$ one also has $y \in P_M$. 
We say that a polytope $K$ is a \emph{$(1+\varepsilon)$-approximation} to $P_M$ if
$P_M \subseteq K \subseteq (1+\varepsilon) P_M$ (note that this only makes sense for monotone polytopes). 
%Due to the monotonicity of $P_M$, 
This is
equivalent to requiring that for each objective function $c \in \setR_{\geq 0}^E$ one has
\begin{equation} \label{eq:approxPolytopeCriterion}
  \max\{ cx \mid x \in P_M \} \leq \max\{ cx \mid x \in K\} \leq (1+\varepsilon) \cdot  \max\{ cx \mid x \in P_M\},
\end{equation}
where we use again monotonicity.

After the publication of the conference version of this work,
Braun and Pokutta~\cite{DBLP:journals/corr/BraunP14} showed the following
extension:
\begin{theorem}
Let $0<\alpha < 1$. 
If $K$ is a polytope with $P_M \subseteq K \subseteq (1 + \frac{\alpha}{n})P_M$, then $\xc(K) \geq 2^{c n}$ with $c := c(\alpha) > 0$ depending on $\alpha$.
\end{theorem}
However, their paper leaves it open how large the extension complexity of a  $(1+\varepsilon)$-approximation 
has to be if $\varepsilon \geq \frac{1}{n}$. For the sake of comparison, suppose in \eqref{eq:MatchingPolytope} we would 
take the odd cut inequalities only 
for $|U| \leq \Theta(\frac{1}{\varepsilon})$, then we would obtain a polytope $K$ 
with $P_M  \subseteq K \subseteq (1+\varepsilon) P_M$ which has only
 $n^{\Theta(1/\varepsilon)} = 2^{\Theta(\frac{\log n}{\varepsilon})}$ many facets\footnote{To see that $K \subseteq (1+\varepsilon) P_M$, 
suppose that $x \in \setR_{\geq 0}^E$ satisfies
$x(\delta(v)) \leq 1$ for $v \in V$ and $x(E(U)) \leq \frac{|U|-1}{2}$ for every odd $U \subseteq V$ of size $|U| \leq k$. Already from the degree constraints we know that for any $U$ we have at least $x(E(U)) \leq \frac{|U|}{2}$.
If we define a slightly scaled vector $y := (1-\frac{1}{k}) \cdot x$, then for any odd $U$ of size $|U| > k$ we have $y(E(U)) \leq (1-\frac{1}{k}) \frac{|U|}{2} \leq \frac{|U|-1}{2}$. Hence $y \in P_M$.}.

We want to argue now that the result of Braun and Pokutta~\cite{DBLP:journals/corr/BraunP14}
implies a lower bound for the whole spectrum of $\varepsilon$:
%implies that this is tight, apart from the $\log n$ term in the exponent:  
%We observe the following (let us stress that $\epsilon$ does not need to be a constant here): 
\begin{corollary}
Let $G = (V,E)$ be the complete graph on $n$ nodes and let $P_M$ be the convex hull of 
all matchings. Let $K$ be a polytope with $P_M \subseteq K \subseteq (1+\varepsilon) P_M$. %and $\varepsilon \geq \frac{2}{n}$. 
Then $\textrm{xc}(K) \geq 2^{\Omega(\min\{ 1/\varepsilon,n\})}$.
\end{corollary}
\begin{proof}
Assume that $\varepsilon \geq \frac{1}{2n}$, otherwise the original result of Braun and Pokutta already applies.
We may also assume that $\frac{1}{\varepsilon}$ is an even integer. 
Let us write $P_M(G)$ to emphasize that we talk about the matching polytope for 
graph $G$. 
Suppose that we have a polytope $K = \{ x \in \setR^E : \exists y: (x,y) \in Q\}$ 
with $P_M(G) \subseteq K \subseteq (1 + \varepsilon) \cdot P_M(G)$ so that $Q$ has  $\xc(K)$ many facets. 
Take any set $V' \subseteq V$ of $k := \frac{1}{2\varepsilon} \leq n$ vertices and let $G' = (V',E')$ be the induced subgraph. 
For a vector $x \in \setR^E$, we write $x = (x',x'')$ with $x' \in \setR^{E'}$ and $x'' \in \setR^{E/E'}$.
Then
\[
 K' := \{ x' \in \setR^{E'} \mid \exists y: ((x',\bm{0}),y) \in Q \}
\]
is a polytope with $\xc(K') \leq \xc(K)$. Intuitively, $K'$ emerges from $K$ 
by ``deleting'' variables for edges that are not in $G'$.  % from the graph that do not run within $V'$, hence 
But $K'$ is still an approximation to the matching polytope for $G'$;
formally $P_M(G') \subseteq K' \subseteq (1+\varepsilon) P_M(G')$.
This can be easily seen by checking criterion~\eqref{eq:approxPolytopeCriterion} 
for objective functions  $c$  with $c_e = 0$ for $e \in E\setminus E'$. %that are $0$ on edges in $E/E'$.
But for this smaller graph, we have that $\varepsilon = \frac{1}{2|V'|}$, hence we can apply the original result of Braun and 
Pokutta to the graph $G'$ and obtain that  $\xc(K') \geq 2^{\Omega(|V'|)}$. This shows the claim. 
\end{proof}
This provides a fairly tight bound on the approximability of the matching polytope whenever $\varepsilon \ll \frac{1}{\log n}$. 
However, the gap between known upper and lower bound remains huge if $\varepsilon$ is a constant.
%It is still open whether the upper bound of $n^{O(1/\varepsilon)}$ on the approximate expension complexity
%is tight, if $\varepsilon$ is in a constant regime. 
%is tig--- except of the probably most interesting range when $\varepsilon$ is a constant.

%\subsection{Comparison to the conference version}
%
%In the original conference version of this paper, the analysis of the contribution of bad pairs $(T,H)$ was 
%done in a way that was closer to the original proof of Razborov~\cite{Disjointness-Razborov90}.
%We believe that the alternative proof that we presented here is the slicker one.  
%Also, the conference version did not contain the observation from Section~\ref{sec:Inapproximability}.

\subsection{Subsequent development}

After a sequence of papers showed lower bounds on the size of linear programs (including \cite{TSP-lower-bound-Fiorini-et-al-STOC2012,LPs-for-CSPs-ChanLeeRaghavendraSteurer-FOCS2013} and this paper), 
the natural next challenge was whether one could also prove lower bounds on the size of
\emph{semidefinite programs}. A recent breakthrough of Lee, Raghavendra and Steurer~\cite{SDPlowerBound-LeeRaghavendraSteurer-STOC2015} %\cite{DBLP:journals/corr/LeeRS14}
answers this affirmatively for the correlation polytope, the
cut polytope and approximate versions of constraint satisfaction problems. However, it is still unknown
whether there is a polynomial size SDP for the perfect matching polytope.

\subsection{Acknowledgements}

The author is grateful to Fritz Eisenbrand, Michel X. Goemans, Jochen K{\"o}nemann and Laura Sanit{\`a}
for helpful discussions and to Sam Fiorini for reading and commenting on a preliminary draft. 
Moreover, the author wants to thank the anonymous reviewers of both the conference and the journal
version for their numerous helpful suggestions and comments. 

\bibliographystyle{alpha}
\bibliography{matchingpolytope}

\end{document}